\documentclass{article}
\usepackage{amsmath,amsfonts,amsthm,amssymb,amscd,color,xcolor}
\usepackage{hyperref}

\colorlet{darkblue}{blue!50!black}

\hypersetup{
    colorlinks,%
    citecolor=darkblue,%
    filecolor=red,%
    linkcolor=darkblue,%
    urlcolor=magenta,%
    pdfnewwindow=true,%
    pdfstartview={FitH}
}

\renewcommand{\Re}{\mathop{\rm Re}\nolimits}

\newcommand{\p}{\partial}
\newcommand{\e}{\varepsilon}

\newcommand{\R}{{\mathbb R}}
\newcommand{\IP}{{\mathbb P}}
\newcommand{\Z}{{\mathbb Z}}
\newcommand{\E}{{\mathbb E}}
\newcommand{\T}{{\mathbb T}}

\newcommand{\HHH}{{\boldsymbol H}}
\newcommand{\VVV}{{\boldsymbol V}}
\newcommand{\III}{{\boldsymbol I}}
\newcommand{\XXX}{{\boldsymbol X}}

\newcommand{\nnu}{{\boldsymbol\nu}}

\newcommand{\BB}{{\mathcal B}}
\newcommand{\DD}{{\mathcal D}}
\newcommand{\CC}{{\mathcal C}}

\newcommand{\FF}{{\mathcal F}}
\newcommand{\HH}{{\mathcal H}}
\newcommand{\KK}{{\mathcal K}}

\newcommand{\PP}{{\mathcal P}}
\newcommand{\RR}{{\mathcal R}}

\newcommand{\TT}{{\mathcal T}}

\newcommand{\PPPP}{{\mathfrak P}}
\newcommand{\pppp}{{\mathfrak p}}

\newcommand{\mmmm}{{\mathfrak m}}

\newcommand{\uuu}{{\boldsymbol{\mathit u}}}
\newcommand{\vvv}{{\boldsymbol{\mathit v}}}

\newcommand{\zzeta}{{\boldsymbol{\zeta}}}
\newcommand{\mmu}{{\boldsymbol{\mu}}}

\newcommand{\bbS}{{\mathbb S}}

\newcommand{\pP}{{\mathbb P}}

\newcommand{\dd}{{\textup d}}

\newcommand{\var}{\mathop{\rm Var}\nolimits}
\newcommand{\supp}{\mathop{\rm supp}\nolimits}
\newcommand{\diver}{\mathop{\rm div}\nolimits}

\newcommand{\esssup}{\mathop{\rm ess\ sup}\nolimits}
\newcommand{\Ent}{\mathop{\rm Ent}\nolimits}
\newcommand{\Ep}{\mathop{\rm Ep}\nolimits}

\theoremstyle{plain}
\newtheorem*{mta}{Theorem A}
\newtheorem*{mtb}{Theorem B}
\newtheorem{theorem}{Theorem}[section]
\newtheorem{lemma}[theorem]{Lemma}
\newtheorem{proposition}[theorem]{Proposition}
\newtheorem{corollary}[theorem]{Corollary}
\theoremstyle{definition}
\newtheorem{definition}[theorem]{Definition}
\newtheorem{condition}[theorem]{Condition}
\theoremstyle{remark}

\numberwithin{equation}{section}

\begin{document}
\author{V.~Jak\v si\'c\footnote{Department of Mathematics and Statistics,
McGill University, 805 Sherbrooke Street West, Montreal, QC, H3A 2K6
Canada; e-mail: \href{mailto:Jaksic@math.mcgill.ca}{Jaksic@math.mcgill.ca}}
\and V. Nersesyan\footnote{Laboratoire de Mat\'ematiques, UMR CNRS 8100, Universit\'e de Versailles-Saint-Quentin-en-Yvelines, F-78035 Versailles, France;  e-mail: \href{mailto:Vahagn.Nersesyan@math.uvsq.fr}{Vahagn.Nersesyan@math.uvsq.fr}}
\and C.-A.~Pillet\footnote{Aix Marseille Universit\'e, CNRS, CPT, UMR 7332, Case 907, 13288 Marseille, France; 
Univerist\'e de Toulon, CNRS, CPT, UMR 7332, 83957 La Garde,  France; e-mail: \href{mailto:pillet@univ-tln.fr}{pillet@univ-tln.fr}}
\and A.~Shirikyan\footnote{Department of Mathematics, University of Cergy--Pontoise, CNRS UMR 8088, 2 avenue Adolphe Chauvin, 95302 Cergy--Pontoise, France; e-mail: \href{mailto:Armen.Shirikyan@u-cergy.fr}{Armen.Shirikyan@u-cergy.fr}}}

\title{Large deviations and Gallavotti--Cohen principle for dissipative PDE's with rough noise}
\date{}
\maketitle

\begin{abstract}
We study a class of dissipative PDE's perturbed by an unbounded kick force. Under some natural assumptions, the restrictions of solutions to integer  times form a homogeneous Markov process. Assuming that the noise is rough with respect to the space variables and has a non-degenerate law, we prove that the system in question satisfies a large deviation principle (LDP) in $\tau$-topology. Under some additional hypotheses, we establish a Gallavotti--Cohen type symmetry for the rate function of an entropy production functional and the strict positivity and finiteness of the mean entropy production rate in the stationary regime. The latter result is applicable to PDE's with strong nonlinear dissipation. 


\end{abstract}

\tableofcontents

\setcounter{section}{-1}

\section{Introduction}
\label{s0}
Let~$H$ be a separable Hilbert space and let $S:H\to H$ be a continuous mapping. We consider a discrete-time Markov process defined by the equation
\begin{equation} \label{0.1}
u_k=S(u_{k-1})+\eta_k, \quad k\ge1,
\end{equation}
where $\{\eta_k\}$ is a sequence of i.i.d.\;random variables in~$H$. This type of systems naturally arise when studying the large-time asymptotics of randomly forced PDE's, and we do not discuss here our motivation, referring the reader to Section~2.3 of the book~\cite{KS-book}. Equation~\eqref{0.1} generates a homogeneous family of Markov chains, and its ergodic theory is well understood in the case when~$S$ possesses a dissipativity property and the law of~$\eta_k$ is sufficiently non-degenerate. Namely, let us assume that 
\begin{equation} \label{0.2}
\|S(u)\|\le q\|u\|+C\quad\mbox{for any $u\in H$},
\end{equation}
where $q<1$ and~$C$ are some numbers not depending on~$u$. If, in addition, the mapping~$S$ is compact in the sense that the image under~$S$ of any bounded set is relatively compact, then the existence of a stationary distribution can easily be proved with the help of the Bogolyubov--Krylov argument. The uniqueness of a stationary measure and its mixing properties are much more delicate questions, and in this paper we deal with a ``rough'' noise, in which case convergence to the unique stationary measure holds in the total variation distance. To describe the problems and results, let us  assume that the law~$\ell$ of the random variables~$\eta_k$ is a Gaussian measure. In this situation, the above-mentioned roughness condition takes the form:
\begin{itemize}
\item[\bf (H)] 
{\sl The mapping $S$ is continuous from~$H$ to the Cameron--Martin space of~$\ell$ and is bounded  on any ball.}
\end{itemize}
Under this hypothesis, the transition probabilities of the Markov family associated with~\eqref{0.1} are all equivalent, and the uniqueness of a stationary measure and its stability in the total variation norm follows from the well-known Doob's theorem; e.g., see Chapter~4 in~\cite{DZ1996}. We refer the reader to the pioneering articles~\cite{yaglom-1947,doob-1948} for first results of this type, to the book~\cite{MT1993} for a general ergodic theory of Markov chains, and to the paper~\cite{BKL-2001} for a proof of the above-mentioned existence and stability result in the case of Navier--Stokes equations on the 2D torus. 

\smallskip
The aim of this paper is twofold: firstly, to establish a large deviation principle (LDP) for occupation measures of~\eqref{0.1} and some physically relevant functionals and, secondly, to derive a Gallavotti--Cohen type symmetry for the rate function corresponding to entropy production. Without going into technical details, we now describe our main results in the case of the 1D Burgers equation on the circle~$\bbS=\R/2\pi\Z$. Namely, let us denote by~$H$ the space of square-integrable functions on~$\bbS$ with zero mean value and consider the problem
\begin{align}
\p_tu-\nu\p_x^2u+u\p_xu&=h(x)+\eta(t,x),\label{0.3}\\
u(0,x)&=u_0(x). \label{0.4}
\end{align}
Here $x\in\bbS$, $\nu>0$ is a parameter, $h\in H$ is a fixed function, and $\eta(t,x)$ is a random process of the form
\begin{equation} \label{0.5}
\eta(t,x)=\sum_{k=1}^\infty \eta_k(x)\delta(t-k),
\end{equation}
where $\{\eta_k\}$ is a sequence of i.i.d.\;Gaussian random variables in~$H$ and~$\delta(t)$ denotes the Dirac measure at zero. Normalising trajectories of~\eqref{0.3} to be right-continuous and denoting $u_k=u(k,x)$, we see that the sequence~$\{u_k\}$ satisfies Eq.~\eqref{0.1}, where~$S:H\to H$ denotes the time-$1$ shift along trajectories of~\eqref{0.3} with $\eta\equiv0$. For any trajectory~$\{u_k\}$, let~$\zzeta_k(u_0)$ be the corresponding occupation measure:
$$
\zzeta_k(u_0)=\frac{1}{k}\sum_{n=0}^{k-1}\delta_{\uuu_n}, 
\quad \uuu_n=(u_l,l\ge n),
$$
where $\delta_\vvv$ denotes the Dirac mass concentrated at~$\vvv=(v_l,l\ge0)$ in the space of probability measures on~$\HHH=H^{\Z_+}$. Thus, $\{\zzeta_k(u_0)\}$ is a sequence of random probability measures on~$\HHH$, and we wish to investigate the problem of large deviations for it. Let us denote by~$V^s$ the space of functions in the Sobolev space of order~$s$ on~$\bbS$ whose mean value is equal to zero. 

\begin{mta}
Let us assume that $h\in V^s$ for an integer $s\ge0$ and the law~$\ell$ of the i.i.d.\;random variables~$\eta_k$ is a centred Gaussian measure on~$H$ such that~$V^{s+1}$ is continuously embedded into its Cameron--Martin space. Then the discrete-time Markov process associated with~\eqref{0.3} has a unique stationary measure~$\mu$, which is exponentially mixing in the sense that the law of any trajectory converges to~$\mu$ in the total variation metric exponentially fast. Moreover, for any initial point~$u_0\in H$, the occupation measures~$\zzeta_k(u_0)$ satisfy the LDP in the $\tau_p$-topology with a good rate function not depending on~$u_0$. 
\end{mta}

The reader is referred to Section~\ref{s1} for the definition of the concepts used in this theorem. We now turn to the question of the Gallavotti--Cohen fluctuation principle. To this end, given a vector $a\in H$, denote by~$\ell_a$ the image of~$\ell$ under the translation in~$H$ by the vector~$a$. The hypotheses of Theorem~A imply that the shifted measure $\ell_{S(u)}$ is equivalent to $\ell$. Thus, the transition kernel of the Markov chain defined by \eqref{0.1} is given by
$P(u,\dd v)=\ell_{S(u)}(\dd v)=\rho(u,v)\ell(\dd v)$,
the density~$\rho(u,v)$ being positive for any $u\in H$ and $\ell$-almost 
every~$v\in H$. This further implies that, for any $k>0$, the law~$\lambda_k$ of
the random variable~$u_k$ is equivalent to~$\ell$,
irrespective of the law~$\lambda_0$ of the initial condition~$u_0$.
In particular, the stationary measure~$\mu$ is equivalent to~$\ell$. We denote by~$\rho$ its density. Thus, when discussing the long time behaviour of the system, we can assume that all the measures~$\lambda_k$ belong to the equivalence class of~$\ell$.

Adapting Gaspard's argument~\cite{gaspard-2004} to our setup, we measure the entropy of the system at 
time $k$ by the relative entropy of $\lambda_k$ with respect to the stationary measure $\mu$:
$$
S(\lambda_k)=\Ent(\lambda_k|\,\mu)
=-\int_H\log\left(\frac{\dd\lambda_k}{\dd\mu}\right)\dd\lambda_k.
$$
We note that the basic properties of relative entropy imply that $S(\lambda)\le0$, with equality if and only if $\lambda_k=\mu$.
The change of entropy in one time step is given by
\begin{equation*}
\delta S(\lambda)=S(\mathfrak{P}_1^\ast\lambda)-S(\lambda),
\end{equation*}
where $\mathfrak{P}_k^\ast$ denotes the Markov semigroup associated with the transition kernel~$P$. Let~$\boldsymbol{\lambda}$ be the law induced on~$\HHH$ by the initial distribution~$\lambda$. Define the following function on $\HHH$:
\begin{equation} \label{0.6}
J(\uuu)=\log\frac{\rho(u_0)\rho(u_0,u_1)}{\rho(u_1)\rho(u_1,u_0)}.
\end{equation}
In the third section of the Appendix, we shall show that, under the hypotheses of Theorem~A,
\begin{equation}\label{EntBalance}
\delta S(\lambda)=\Ep(\lambda)
-\int_\HHH J(\uuu)\boldsymbol{\lambda}(\dd\uuu),
\end{equation}
where the functional\footnote{For a precise definition of~$\Ep(\lambda)$, see~\eqref{EpForm}.} 
$\Ep(\,\cdot\,)$ is such that $\Ep(\lambda)\ge0$ for
all $\lambda$ in the equivalence class of $\ell$. Moreover, $\Ep(\lambda)=0$
if and only if  $\lambda=\mu$ and~$\mu$ satisfies the detailed balance condition
\begin{equation}\label{DBCond}
\rho(u)\rho(u,v)=\rho(v)\rho(v,u),
\end{equation}
$\ell\otimes\ell$-almost everywhere on $H\times H$. The validity of  Eq.~\eqref{DBCond} is well known to
be necessary and sufficient to ensure the time-reversal invariance of the Markov chain under the
stationary law $\boldsymbol{\mu}$. The functional ${ \rm Ep}(\,\cdot\,)$ is thus a
measure of the breakdown of time-reversal invariance, a phenomenon usually connected with
the production of entropy. We shall therefore identify $\Ep(\lambda)$ with the entropy production rate
of the system in the state $\lambda$.
Reading Eq.~\eqref{EntBalance} as an entropy balance relation, we may consequently interpret the observable
$J$ as the entropy dissipated into the environment, i.e., the integral of the outgoing entropy flux over the unit time interval. Note that the vanishing of the entropy flux observable $J$ is equivalent to the detailed balance condition~\eqref{DBCond}. We shall prove in Section \ref{s9.3} that the unique stationary measure $\mu$ does 
not satisfy the detailed balance relation, so that $\Ep(\lambda)>0$ for all $\lambda$.

In terms of the random variables
\begin{equation}\label{1.14}
\xi_k(\uuu)=\frac1k\sum_{n=0}^{k-1}\sigma(u_n,u_{n+1}),
\end{equation}
where
\begin{equation} \label{entropy}
\sigma(u,v)=\log\frac{\rho(u,v)}{\rho(v,u)},
\end{equation}
we can write the entropy balance relation over $k$ time steps as
\begin{align}
\frac1k(S(\lambda_k)-S(\lambda_0))=\frac1k\sum_{n=0}^{k-1}\Ep(\mathfrak{P}_n^*\lambda_0)
&-\int_{\HHH}\xi_k(\uuu)\boldsymbol{\lambda}(\dd\uuu)\label{kBal}\\
&+\frac1k\int_H\log\rho(u)(\lambda_k(\dd u)-\lambda_0(\dd u)).\nonumber
\end{align}
The last term on the right hand side of this relation (a so-called boundary term) becomes negligible in the large 
time limit. It vanishes in the stationary regime where the previous relation becomes
\begin{equation}\label{0.7}
\Ep(\mu)=\int_\HHH\xi_k(\uuu)\mmu(\dd \uuu)
=\int_{H\times H}\rho(u)\rho(u,v)\sigma(u,v)\,\ell(\dd u)\ell(\dd v).
\end{equation}
In the third subsection of the Appendix, we shall briefly discuss the relation of the observable~$\sigma$ with time-reversal of the path measure~$\boldsymbol{\mu}$ and its connection with dynamical (Kolmogorov--Sinai) entropy.

According to Eq.~\eqref{0.7}, the mean entropy flux is non-negative. By the law of large numbers, the sequence~$\xi_k$ converges $\boldsymbol{\mu}$-a.s.\;towards~$\Ep(\mu)$. The Gallavotti--Cohen fluctuation relation is a statement about the large deviations of~$\xi_k$ from this limit. Roughly speaking, it says that
$$
\frac{\boldsymbol{\mu}\left(\xi_k\simeq -r\right)}{\boldsymbol{\mu}\left(\xi_k\simeq +r\right)}
\simeq\mathrm{e}^{-kr}\quad\mbox{for large $k$}.
$$

The fact that the entropy production rate is non-negative and the definition of 
the entropy flux observable $\sigma$ are part of the general theory 
of non-equilibrium  statistical mechanics in the mathematical framework of deterministic and stochastic dynamical systems~\cite{ES-1994,GC-1995,ruelle-1997,ruelle-1999,kurchan-1998,maes-1999,gaspard-2004,RM-2007,JPR-2011}.  On the other hand, detailed dynamical questions like strict positivity of the entropy production rate, LDP for the entropy flux, and validity of the Gallavotti--Cohen fluctuation relation can be answered only in the context of concrete models.
In some cases, it is possible to relate the observable~$\sigma$ to the fluxes of some physical quantities, typically
heat or some other forms of energy. In this respect, we refer the reader to~\cite{BM-2005} for the discussion of a closely related model. In this paper, we shall prove the following result. 

\begin{mtb}
In addition to the hypotheses of Theorem~A, let us assume that $h\in V^{2s+1}$ and the set of normalised eigenvectors of the covariance operator for~$\ell$ coincides with the trigonometric basis in~$H$. Then, for any initial condition $u_0\in H$, the laws of the random variables \eqref{1.14}
satisfy the LDP with a good rate function $I:\R\to[0,+\infty]$ not depending on~$u_0$. Moreover, the entropy production rate is strictly positive, 
\begin{equation} \label{1.16}
\Ep(\mu)=\int_{\HHH}\sigma(u_0,u_1)\boldsymbol{\mu}(\dd\uuu)>0,
\end{equation}
and the Gallavotti--Cohen fluctuation relation\,\footnote{Relation~\eqref{1.15} means, in particular, that $I(r)=+\infty$ if and only if $I(-r)=+\infty$.} holds for~$I$:
\begin{equation} \label{1.15}
I(-r)=I(r)+r\quad \mbox{for $r\in\R$}.
\end{equation}
\end{mtb}

There is an enormous literature on mathematical, physical, numerical, and experimental aspects of  Gallavotti--Cohen fluctuation relation (some of the references can be found in~\cite{JPR-2011, RM-2007}). The previous mathematically rigorous works  closest to ours are~\cite{LS-1999,EPR-1999b,EPR-1999a,EH-2000,RT-2002}. Lebowitz and Spohn~\cite{LS-1999}, building  on the previous work by  Kurchan~\cite{kurchan-1998}, have developed a general theory of Gallavotti--Cohen fluctuation relations for finite-dimensional Markov processes with applications to various models, including diffusion and simple exclusion processes. In~\cite{EPR-1999b,EPR-1999a,EH-2000}, the authors consider a finite anharmonic chain coupled to two thermal reservoirs at its ends. Its analysis reduces to a study of suitable finite-dimensional Markov process with degenerate noise. In particular, the local Gallavotti--Cohen fluctuation relation for this model has been established in~\cite{RT-2002}. To the best of our knowledge, there were no previous mathematically rigorous studies of Gallavotti--Cohen fluctuation relation  for nonlinear PDE's  driven by a stochastic forcing (note, however, that the LDP for the Navier--Stokes and Burgers equations was proved in the papers~\cite{gourcy-2007a,gourcy-2007b} for the case of a rough white-noise force and in~\cite{JNPS-2012} for the case of a smooth bounded kick force).  On the physical level of rigour, Maes and coworkers~\cite{MRV-2001,MN-2003,maes-2004} have examined in depth the fluctuation relation for stochastic dynamics. In a somewhat different spirit, inspired by the thermodynamic formalism of dynamical systems, we should also mention the works of Gaspard~\cite{gaspard-2004} and Lecomte {\sl et al.}~\cite{LAW-2007}.

The LDP for the Burgers equation stated in Theorem~A is true for other more complicated models, such as the Navier--Stokes system or the complex Ginzburg--Landau equation, while the Gallavotti--Cohen fluctuation relation of
Theorem~B  remains valid for problems with strong nonlinear dissipation, such as the reaction--diffusion system with superlinear interaction. Moreover, the law of~$\eta_k$ does not need to be Gaussian, and the results we prove are true for a rather general class of decomposable measures;  see Sections~\ref{s1} and~\ref{s9} for details. As for the positivity of the mean entropy production (which is equivalent to the absence of the detailed balance~\eqref{DBCond}), it uses the unboundedness of the phase space, continuity of the transition densities, and a particular structure of the density of the random perturbation, well suited for applying Laplace-type asymptotics for integrals. This type of argument seems to be new in the context considered in our paper. 

The somewhat surprising fact that the global LDP for unbounded observables holds for the Burgers  and reaction-diffusion equations has its physical  origin in the strong dissipation characterising   these models. It is natural to expect that in more generic situations (like Navier--Stokes systems) only a local LDP and, hence, a local fluctuation relation holds (like in~\cite{RT-2002}). However, in the absence of a strong dissipative mechanism, our method of proof of LDP for occupational measures is not suited for establishing local LDP for unbounded observables like the entropy flux. It is likely that  more specific techniques that 
deal directly with LDP for the entropy flux are needed to analyse this question. We plan to address this problem in future publications. 

Finally, let us mention that our technique for investigating the LDP for an entropy production functional is based on the following two properties: a)~dissipativity and parabolic\footnote{By the parabolic regularisation, we mean the property of exponential stability, with an arbitrarily large rate, after removing finitely many modes.} regularisation for the underlying PDE; b)~finite smoothness of the noise. Property~a) is not satisfied, for instance, in the case of a damped nonlinear wave equation, and it is an interesting open question to extend our results to that situation. As for~b), it is crucial for the very definition of the entropy production, and the case of infinitely smooth noise remains out of reach. 

\smallskip
The paper is organised as follows. In Section~\ref{s1}, we formulate our main abstract results on the large deviations  and the Gallavotti--Cohen fluctuation theorem. Various applications of these results are discussed in Section~\ref{s9}. Sections~\ref{s2} and~\ref{s11} are devoted to proving the  theorems announced in Section~\ref{s1}. The Appendix gathers some auxiliary results  on decomposable measures and LDP for Markov chains and discusses the analogy between our models and heat conducting networks. 

\subsection*{Acknowledgments} 
The authors are grateful to A.~Boritchev whose careful reading of the manuscript helped to improve the presentation and to eliminate a number of misprints. The research of VJ was supported by NSERC.
The research of VN was supported by the ANR grants EMAQS (No. ANR 2011 BS01 017~01) and STOSYMAP (No.~ANR 2011 BS01 015 01). The research of CAP was partly supported by ANR grant 09-BLAN-0098. The research of AS was carried out within the MME-DII Center of Excellence (ANR-11-LABX-0023-01) and supported by the ANR grant STOSYMAP and RSF research project 14-49-00079. This paper was finalised when AS was visiting the Mathematics and Statistics Department of the University of McGill, and he thanks the institution for hospitality and excellent working conditions. 

\subsection*{Notation}
Let $X$ be a Polish space with a metric~$d$. We always assume that it is endowed with its Borel $\sigma$-algebra
$\BB_X$. Given $R>0$ and $a\in X$, we denote by $B_X(a,R)$ the closed ball in~$X$ of radius~$R$ centred at~$a$. The following spaces are systematically used in the paper.

\smallskip
\noindent
$\XXX=X^{\Z_+}$ denotes the direct product of countably many copies of~$X$. The space~$\XXX$ is endowed with the Tikhonov topology, and its elements are denoted by~$\uuu=(u_n,n\ge0)$. We write $X^m$ for the direct product of~$m$ copies of~$X$. 

\smallskip
\noindent
$C(X)$ is the space of continuous functions $f:X\to\R$. We denote by~$C_b(X)$ the subspace of bounded functions in~$C(X)$ and  endow it with the natural norm $\|f\|_\infty=\sup_X|f|$. 

\smallskip
\noindent
$\PP(X)$ denotes the space of probability measures on~$X$.  Given $\mu\in\PP(X)$ and a $\mu$-integrable function~$f:X\to\R$, we write
$$
\langle f,\mu\rangle=\int_Xf(u)\mu(\dd u). 
$$
The total variation metric on~$\PP(X)$ is defined by
\begin{align*}
\|\mu_1-\mu_2\|_{\mathrm{var}}=\frac12\sup_{\|f\|_\infty\le 1}
|\langle f,\mu_1\rangle-\langle f,\mu_2\rangle|
=\sup_{\Gamma\in\BB_X}|\mu_1(\Gamma)-\mu_2(\Gamma)|.
\end{align*}

\smallskip
\noindent
$C(J,H)$ denotes the space of continuous functions on an interval~$J\subset\R$ with range in the Banach space~$H$. We write $C_b(J,X)$ for the subspace of bounded functions and endow it with the natural norm
$$
\|f\|_{L^\infty(J,H)}=\esssup\limits_{t\in J}\|f(t)\|_H. 
$$

\smallskip
\noindent
$L^p(J,H)$ stands for the space of Borel-measurable functions $f:J\to H$ such that
$$
\|f\|_{L^p(J,H)}=\biggl(\int_J\|f(t)\|_H^p\dd t\biggr)^{1/p}<\infty.
$$
In the case $p=\infty$, the above norm should be replaced by $\|f\|_{L^\infty(J,H)}$. 

\smallskip
\noindent
We denote by~$C, C_1,C_2,\dots$ unessential positive numbers. 

\section{Main results}
\label{s1}
In this section, we introduce a class of discrete-time Markov processes and formulate a result on the existence, uniqueness, and exponential mixing of a stationary measure and the large deviation principle for the occupation measures and some unbounded functionals. We next discuss the Gallavotti--Cohen fluctuation theorem for an entropy production functional.  

\subsection{The model}
\label{s1.1}
Let~$H$ be a separable Hilbert space, let $S:H\to H$ be a continuous mapping, and let~$\{\eta_k,k\ge1\}$ be a sequence of i.i.d.\;random variables in~$H$. We consider the stochastic system~\eqref{0.1}, supplemented with the initial condition
\begin{equation} \label{1.2}
u_0=u\in H. 
\end{equation}
Let us denote by~$(u_k,\IP_u)$ the Markov  family generated by~\eqref{0.1}, \eqref{1.2}, by~$P_k(u,\Gamma)$ its transition function, and by $\PPPP_k:C_b(H)\to C_b(H)$ and $\PPPP_k^*:\PP(H)\to\PP(H)$ the corresponding Markov semigroups. Given a measure $\lambda\in\PP(H)$, we write $\IP_\lambda(\cdot)=\int_H\IP_u(\cdot)\lambda(\dd u)$. We shall always assume that~$S$ satisfies the two conditions below.

\medskip
{\bf (A) Continuity and compactness.} 
{\sl There is a separable Banach space~$U$ compactly embedded into~$H$ such that~$S$ is continuous from~$H$ to~$U$ and is bounded on any ball.}

\smallskip
{\bf (B) Dissipativity.} 
{\sl There is a continuous function $\varPhi:H\to\R_+$ bounded on any ball and such that $\varPhi(u)\to+\infty$ as $\|u\|\to+\infty$ and
\begin{equation} \label{1.1}
\varPhi(S(u)+v)\le q\,\varPhi(u)+C(\varPhi(v)+1)\quad\mbox{for all $u,v\in H$},
\end{equation}
where $q<1$ and $C\ge1$ do not depend on~$u$ and~$v$.}

\smallskip
As for the random variables~$\{\eta_k\}$, we assume that their law has a particular structure related to~$S$. To formulate this condition, we shall use some concepts defined in Section~\ref{s8.1}. Given a vector $a\in H$ and a measure $\ell\in\PP(H)$, we denote by $\theta_a:H\to H$ the shift operator in~$H$ taking~$u$ to~$u+a$, by $\ell_a=\ell\circ\theta_a^{-1}$ the image of~$\ell$ under~$\theta_a$, and by~$H_\ell$ the set of all admissible shifts for~$\ell$. 

\medskip
{\bf (C) Structure of the noise.}
{\sl The support of the measure~$\ell:=\DD(\eta_1)$ coincides with~$H$, and there is $\delta>0$ such that
\begin{equation} \label{1.07}
\mmmm_\delta(\ell):=\int_He^{\delta \varPhi(u)}\,\ell(\dd u)<\infty. 
\end{equation}
Moreover, the Banach space~$U$ defined in~{\rm(A)} is contained in the semigroup of admissible shifts~$H_{\ell}$, and the mapping $\theta:U\to\PP(H)$ that takes~$a\in U$ to~$\ell_a$ is continuous, provided that the space~$\PP(H)$ is endowed with the total variation norm.} 

\medskip
A sufficient condition for the validity of some of the above properties is given in Proposition~\ref{p5.3}. In the next two subsections, we formulate our main results on the exponential mixing, the LDP in the space of trajectories (or level-$3$ LDP), and the Gallavotti--Cohen fluctuation relation. 

\subsection{Exponential mixing and large deviations}
\label{s1.2}
For the reader's convenience, we begin with some well-known definitions. Let~$X$ be a topological space, endowed with its Borel $\sigma$-algebra $\BB_X$, and let~$\PP(X)$ be the set of probability measures on~$X$, which is endowed with a regular\footnote{Recall that a topological space~$(Y,\TT)$ is said to be {\it regular\/} if any singleton is a closed subset, and for any closed set $F\subset X$ and any point $x\notin F$ there are disjoint open subsets~$G_1$ and~$G_2$ such that $F\subset G_1$ and $x\in G_2$.} topology~$\TT$ and the corresponding Borel $\sigma$-algebra. Recall that a mapping $I:\PP(X)\to[0,+\infty]$ is called a {\it rate function\/} if it is lower semicontinuous, and a rate function~$I$ is said to be {\it good\/} if its level sets are compact. For a Borel subset~$\Gamma\subset\PP(X)$, we write~$I(\Gamma)=\inf_{\sigma\in\Gamma}I(\sigma)$.

Now let $\{\zeta_k\}$ be a sequence of random probability measures\,\footnote{This means that $\zeta_k$ is a measurable mapping from~$(\Omega,\FF)$ with range in the space~$\PP(X)$.} on~$X$ defined on a measurable space~$(\Omega,\FF)$, let  $\Lambda$ be an arbitrary set, and let~$\IP_\lambda$ be a family of probabilities on~$(\Omega,\FF)$ indexed by $\lambda\in\Lambda$. 

\begin{definition}
We shall say that~$\{\zeta_k\}$ satisfies the {\it uniform LDP\/} with $\lambda\in \Lambda$ and a rate function~$I$ if the following two properties hold. 
\begin{description}
\item[Upper bound.] 
For any closed subset~$F\subset\PP(X)$, we have
$$
\limsup_{k\to\infty} \frac1k\log \sup_{\lambda\in \Lambda} \pP_\lambda\{\zeta_k\in F\}\le -I(F).
$$
\item[Lower bound.] 
For any open subset~$G\subset\PP(X)$, we have
$$
\liminf_{k\to\infty} \frac1k\log \inf_{\lambda\in \Lambda}\pP_\lambda\{\zeta_k\in G\}\ge -I(G).
$$
\end{description}
\end{definition}   

We now consider a particular case in which~$X$ is the product space $\HHH=H^{\Z_+}$, endowed with the Tikhonov topology. For any integer~$k\ge1$, consider the space~$\PP(H^k)$ endowed with the $\tau$-topology, which is defined as the weakest topology with respect to which all the functionals $\mu\mapsto (f,\mu)$ with $f\in L^\infty(H^k)$ are continuous. We shall write~$\PP_\tau(H^k)$ to emphasise the $\tau$-topology on~$\PP(H^k)$. The space of probability measures $\PP(\HHH)$ is endowed with the projective limit topology~$\tau_p$ of the system $\{\PP_\tau(H^k),k\ge1\}$. In other words, $\tau_p$ is the weakest topology on~$\PP(\HHH)$ with respect to which all the functionals $\mu\mapsto (f,\mu)$ with $f\in L^\infty(H^k)$ and any $k\ge1$ are continuous. 

Let us  go back to system~\eqref{0.1}. Recall that a measure~$\mu\in\PP(H)$ is said to be {\it stationary\/} for a Markov family $(u_k,\IP_u)$ if $\PPPP_1^*\mu=\mu$. We denote by~$\zzeta_k$ the occupation measure in the trajectory space for a solution of~\eqref{0.1}; that is,
\begin{equation} \label{1.12}
\zzeta_k=\frac1k\sum_{n=0}^{k-1}\delta_{\uuu_n},
\end{equation}
where $\uuu_n=(u_l,l\ge n)$, and~$\{u_l\}$ is a trajectory of~\eqref{0.1}. 
The following theorem establishes uniqueness and mixing of a stationary measure for the Markov family associated with~\eqref{0.1} and a uniform LDP for~$\zzeta_k$ in the $\tau_p$-topology. Its proof is given in Section~\ref{s2}.

\begin{theorem} \label{t1.1}
Let Hypotheses~{\rm(A)}, {\rm(B)}, and~{\rm(C)} be fulfilled and let $(u_k,\IP_u)$ be  the Markov family associated with~\eqref{0.1}. Then~$(u_k,\IP_u)$ has a unique stationary measure~$\mu$, and there are positive numbers~$\gamma$ and~$C_1$ such that
\begin{equation} \label{1.7}
\|\PPPP_k^*\lambda-\mu\|_{\mathrm{var}}
\le C_1e^{-\gamma k} \biggl(1+\int_H \varPhi(u)\,\lambda(\dd u)\biggr)
\quad\mbox{for any $\lambda\in\PP(H)$, $k\ge0$}. 
\end{equation}
Moreover, for any $c>0$ and any subset $\Lambda\subset \PP(H)$ satisfying the condition 
\begin{equation} \label{1.8}
\sup_{\lambda\in\Lambda}\int_He^{c \,\varPhi(u)}\lambda(\dd u)<\infty,
\end{equation}
the uniform LDP with $\lambda\in\Lambda$ and a good rate function $\III:\PP(\HHH)\to[0,+\infty]$ holds  for the sequence of $\IP_\lambda$-occupation measures $\{\zzeta_k,k\ge1\}$. 
\end{theorem}

\smallskip
Theorem~\ref{t1.1} combined with an approximation argument enables one to establish the LDP for various functionals of trajectories of~\eqref{0.1} with moderate growth at infinity. To formulate the corresponding result, we shall need the concept of a stabilisable functional. 

Let $\pppp:H\to[0,+\infty]$ be a lower semicontinuous function. We shall say that~$\pppp$ is {\it uniformly stabilisable\/} for the Markov family~$(u_k,\IP_u)$ if there is an increasing continuous function $Q:\R_+\to\R_+$ and a positive number~$\gamma$ such that
\begin{equation} \label{1.9}
\E_u\exp\bigl(\pppp(u_1)+\cdots+\pppp(u_{k})\bigr)\le Q(\|u\|)e^{\gamma k}\quad\mbox{for $k\ge1$, $u\in H$}. 
\end{equation}

\begin{theorem} \label{t1.2}
Under the hypotheses of Theorem~\ref{t1.1}, let~$\pppp$ be a uniformly stabilisable functional, let $m\ge0$ be an integer, and let $f:H^{m+1}\to\R$ be a measurable function that is bounded on any ball and satisfies the condition
\begin{equation} \label{1.11}
\frac{|f(v_0,\dots,v_m)|}{\pppp(v_0)+\cdots+\pppp(v_m)}\to0\quad
\mbox{as $\|v_0\|+\cdots+\|v_m\|\to+\infty$}.
\end{equation}
Then, for any measure $\lambda\in\PP(H)$ satisfying the condition
\begin{equation} \label{1.17}
\int_H\bigl(\exp\bigl(c\,\varPhi(u)\bigr)+e^{\pppp(u)}Q(\|u\|)\bigr)\lambda(\dd u)<\infty,
\end{equation}
with some $c>0$, the $\IP_\lambda$-laws of the real-valued random variables
$$
\xi_k=\frac{1}{k}\sum_{n=0}^{k-1}f(u_n,\dots,u_{n+m}), \quad k\ge1,
$$
satisfy the LDP with a good rate function $I_f:\R\to[0,+\infty]$ not depending on~$\lambda$.
\end{theorem}

Theorems~\ref{t1.1} and~\ref{t1.2} are applied in Section~\ref{s9} to prove the LDP for various dissipative PDE's with random perturbations. In the next subsection, we discuss a symmetry property of the rate function for a particular choice of the observable~$f$. 

\subsection{Gallavotti--Cohen fluctuation relation}
\label{s1.4}
The entropy flux observable for a general Markov family in~$H$ is defined by~\eqref{0.6}, provided that the transition function $P_1(u,\dd v)$ possesses a density with respect to a reference measure~$\ell\in\PP(H)$,
\begin{equation} \label{1.012}
P_1(u,\dd v)=\rho(u,v)\ell(\dd v),
\end{equation}
and that~$\rho(u,v)>0$ for $\ell\otimes\ell$-almost every $(u,v)$. 
If~$(u_k,\IP_u)$ is the Markov family associated with~\eqref{0.1}, then the existence of a density with respect to the law of~$\eta_k$ follows from Conditions~(A) and~(C), while a sufficient condition for its positivity on a set of full measure is given by Proposition~\ref{p5.3}. By the parameter version of the Radon--Nikodym theorem (see~\cite{novikov-2005}), if~$(u_k,\IP_u)$ possesses the Feller property, then one can choose~$\rho$ to be a measurable function in~$(u,v)$. Given a stationary distribution~$\mu$ of~$(u_k,\IP_u)$, we denote by~$\mmu$ the corresponding path measure and note that~$\mu$ is absolutely continuous with respect to~$\ell$, with the corresponding density given by
\begin{equation} \label{1.011}
\rho(v)=\int_H\rho(z,v)\mu(\dd z). 
\end{equation}
It is straightforward to check that $\rho(v)>0$ for $\ell$-almost every $v\in H$. Recall that the entropy production functional~$\sigma$ is defined by~\eqref{entropy}. We have the following simple result.

\begin{lemma} \label{l1.4}
Let $(u_k,\IP_u)$ be a Feller family of  discrete-time Markov processes in~$H$ such that~\eqref{1.012} holds for a reference measure~$\ell\in\PP(H)$ and a measurable function~$\rho(u,v)$ that is positive $\ell\otimes\ell$-almost everywhere. Let~$\mu$ be a stationary measure of $(u_k,\IP_u)$ such that
\begin{equation} \label{1.010}
\int_H|\log\rho(v)|\,\mu(\dd v)<\infty.
\end{equation}
Then the negative part of~$\sigma$ is $\mmu$-integrable, and the mean value of~$\sigma$ with respect to~$\mmu$ is non-negative. 
\end{lemma}

\begin{proof}
We only need to prove the $\mmu$-integrability of the negative part of~$\sigma$ (which implies in particular that $\langle\sigma\rangle_\mu$, the mean value of~$\sigma$ with respect to~$\mmu$, is well defined), because the non-negativity of~$\langle\sigma\rangle_\mu$   follows immediately from~\eqref{0.7} and the fact that $\Ep(\mu)\ge0$. To this end, setting $\rho_{01}=\rho(v_0,v_1)$ and $\rho_{10}=\rho(v_1,v_0)$ and defining~$\rho$ to be the density of~$\mu$ against~$\ell$, we write
\begin{align*}
\int_{\HHH}\sigma^-\,\dd \mmu
&=\int_{H^2}I_{\{\rho_{01}\le\rho_{10}\}}\Bigl|\log\frac{\rho_{01}}{\rho_{10}}\Bigr|\,P(\dd v_0,\dd v_1)\\
&=\int_{H^2}I_{\{\rho_{01}\le\rho_{10}\}}\log\frac{\rho_{10}}{\rho_{01}}\,P(\dd v_0,\dd v_1)\\
&\le\int_{H^2}I_{\{\rho_{01}\le\rho_{10}\}}\Bigl(\log\frac{\rho_1\rho_{10}}{\rho_0\rho_{01}}-\log\frac{\rho_1}{\rho_0}\Bigr)P(\dd v_0,\dd v_1),
\end{align*}
where $\rho_i=\rho(v_i)$ and $P(\dd v_0,\dd v_1)=\rho_0\rho_{01}\ell(\dd v_0)\ell(\dd v_1)$. Using the inequality $\log x\le x$ for $x>0$, we see that the right-hand side of this inequality does not exceed
$$
\int_{H^2}\rho_1\rho_{10}\ell(\dd v_0)\ell(\dd v_1)
+\int_{H^2}\bigl(|\log\rho_0|+|\log\rho_1|\bigr)P(\dd v_0,\dd v_1). 
$$
The first term of this expression is equal to~$1$, while the second is finite in view of~\eqref{1.010}. 
\end{proof}

We now go back to the Markov family $(u_k,\IP_u)$ associated with~\eqref{0.1} and assume that Conditions~(A)--(C) are fulfilled. Furthermore, we make the following hypothesis:

\smallskip
{\bf (D) Entropy production.} 
{\sl The densities~$\rho(u,v)$ can be chosen so that the observable~$\sigma(v_0,v_1)$ for $(u_k,\IP_u)$  is well defined and bounded on any ball of~$H\times H$. Moreover, there is a uniformly stabilisable functional $\pppp:H\to[0,+\infty]$ such that 
\begin{equation} \label{1.13}
\frac{|\sigma(v_0,v_1)|}{\pppp(v_0)+\pppp(v_1)}\to0\quad
\mbox{as $\|v_0\|+\|v_1\|\to+\infty$}.
\end{equation}
}%
The following theorem establishes the LDP for the entropy production functional calculated on trajectories and the Gallavotti--Cohen fluctuation principle for the corresponding rate function.  

\begin{theorem} \label{t1.5}
Let us assume that Conditions~{\rm(A)--(D)} are fulfilled. Then, for any initial measure $\lambda\in\PP(H)$ satisfying~\eqref{1.17},
the LDP with a good rate function $I:\R\to[0,+\infty]$, independent of~$\lambda$, holds for the $\IP_\lambda$-laws of the real-valued random variables~\eqref{1.14}. Moreover, if~\eqref{1.17} is satisfied for~$\lambda=\ell$, then the Gallavotti--Cohen fluctuation relation~\eqref{1.15} holds for~$I$.
\end{theorem}

A proof of Theorem~\ref{t1.5} is presented in Section~\ref{s11}, and its applications are discussed in Sections~\ref{s9.3} and~\ref{s9.4}.

\section{Applications}
\label{s9}
In this section, we discuss some applications of the results of the foregoing section  to various dissipative PDE's perturbed by an unbounded kick force. We first prove that the hypotheses of Theorems~\ref{t1.1} and~\ref{t1.2} are satisfied for the 2D Navier--Stokes system and the complex Ginzburg--Landau equation. We next show that, in the case of equations with strong damping (such as the Burgers equation with periodic boundary conditions or a reaction-diffusion system with superlinear interaction), Theorem~\ref{t1.5} is also applicable. 

\subsection{Two-dimensional Navier--Stokes system}
\label{s9.1}
We consider the Navier--Stokes system on the torus $\T^2\subset\R^2$. Let us denote by~$\dot L^2$ the space of square-integrable vector fields on~$\T^2$ with zero mean value, introduce the space
\begin{equation} \label{2.0}
H=\bigl\{u\in \dot L^2:\diver u=0\mbox{ on $\T^2$}\bigr\},
\end{equation}
and write~$\Pi$ for the orthogonal projection in~$\dot L^2$ onto~$H$. Restricting ourselves to solutions and external forces with zero mean value with respect to the space variables and projecting the Navier--Stokes system onto~$H$, we obtain the nonlocal evolution equation
\begin{equation} \label{2.1}
\p_tu+\nu Lu+B(u)=f(t).
\end{equation}
Here $\nu>0$ is a parameter, $L=-\Delta$, $B(u)=\Pi(\langle u,\nabla\rangle u)$ is the nonlinear term, and~$f$ is an external force of the form
\begin{equation} \label{2.2}
f(t)=h+\sum_{k=1}^\infty\eta_k\delta(t-k),
\end{equation}
where $h\in H$ is a deterministic function, $\delta(t)$ is the Dirac mass at zero, and~$\{\eta_k\}$ is a sequence of i.i.d.\;random variables in~$H$. Normalising solutions of~\eqref{2.1}, \eqref{2.2} to be right-continuous and setting $u_k=u(k)$, we obtain relation~\eqref{0.1}, in which $S:H\to H$ stands for the time-one shift along trajectories of Eq.~\eqref{2.1} with $f=h$. We recall that~$L$ is a positive self-adjoint operator in~$H$ with a compact inverse and denote by~$\{e_j\}$ an orthonormal basis in~$H$ composed of the eigenfunctions of~$L$, with the eigenvalues~$\{\alpha_j\}$ indexed in a non-decreasing order. Let~$V^s$ be the domain of the operator~$L^{s/2}$, so that $V^s=H^s\cap H$, where~$H^s$ is the Sobolev space of order~$s$ on~$\T^2$. 

The family of all trajectories for~\eqref{0.1} form a discrete-time Markov process, which will be denoted by~$(u_k,\IP_u)$; see Section~2.3 in~\cite{KS-book} for details. We now make the following hypothesis on the stochastic part of the external force~\eqref{2.2}. 

\begin{condition} \label{c2.1}
{\sl The i.i.d.\;random variables~$\eta_k$ have the form {\rm(}cf.~\eqref{8.2}{\rm)}
\begin{equation} \label{2.3}
\eta_k=\sum_{j=1}^\infty b_j\xi_{jk} e_j,
\end{equation}
where $\{b_j\}$ is a sequence of positive numbers such that
\begin{equation} \label{2.4}
\sum_{j=1}^\infty b_j^2<\infty,
\end{equation}
and $\{\xi_{jk}\}$ is a sequence of independent scalar random variables whose laws possess densities~$\tilde\rho_j\in C^1$ with respect to the Lebesgue measure, which are positive everywhere and satisfy~\eqref{8.6} and~\eqref{10.15}.}
\end{condition}

Let us note that if the laws of~$\xi_{jk}$ are centred Gaussian measures with variances~$\sigma_j^2$ belonging to a bounded interval separated from zero, then~\eqref{8.6} and~\eqref{10.15} are satisfied. The following result establishes the LDP for the occupation measures of~$(u_k,\IP_u)$.

\begin{theorem} \label{t2.2}
Let~$s\ge2$ be an integer, let $h\in V^s$, and let~$\eta_k$ be random variables for which Condition~\ref{c2.1} is fulfilled. Assume, in addition, that the law~$\ell$ of~$\eta_k$ satisfies~\eqref{1.07} with $\varPhi(u)=\|u\|$ and some~$\delta>0$, and
\begin{equation} \label{2.5}
\sum_{j=1}^\infty b_j^{-2}\alpha_j^{-1-s}<\infty.
\end{equation}
Then $(u_k,\IP_u)$ has a unique stationary measure~$\mu\in\PP(H)$, which is exponentially mixing in the sense that inequality~\eqref{1.7} holds. Moreover, for  any $c>0$ and any subset $\Lambda\subset\PP(H)$ satisfying~\eqref{1.8}, 
the uniform LDP with $\lambda\in\Lambda$ and a good rate function $\III:\PP(\HHH)\to[0,+\infty]$ holds for the sequence of $\IP_\lambda$-occupation measures~\eqref{1.12}. 
\end{theorem}

Inequality~\eqref{2.5} prevents the random kicks~$\eta_k$ from being very regular functions of~$x$. Indeed, it is well known that $\alpha_j\sim j$ as $j\to\infty$; see~\cite{metivier-1978}. Hence, if  $b_j=j^{-r}$ for $j\ge1$, then the above theorem is applicable only for $r\in(1/2,s/2)$, so that the regularity of~$\eta_k$ is lower than~$V^{s-1}$. Furthermore, by the Cauchy-Schwarz inequality, we have 
$$
+\infty=\sum_{j=1}^\infty \alpha_j^{-1}
\le\biggl(\,\sum_{j=1}^\infty b_j^2\biggr)^{1/2}\biggl(\,\sum_{j=1}^\infty b_j^{-2}\alpha_j^{-2}\biggr)^{1/2}. 
$$ 
If $s\le 1$, then~\eqref{2.4} and~\eqref{2.5} imply that the right-hand side of this inequality is finite. Since~$s$ is an integer, we see that it must satisfy the inequality~$s\ge2$. On the other hand, we claim that if $s\ge2$, then the hypotheses of Theorem~\ref{t2.2} are fulfilled for any 
i.i.d.\;sequence~$\{\eta_k\}$ of Gaussian random variables in~$H$ whose covariance operator is diagonal in the basis~$\{e_j\}$ and has eigenvalues~$\{b_j^2\}$ satisfying~\eqref{2.4} and~\eqref{2.5}. Indeed, a Gaussian measure~$\ell$ is representable as the direct product of its projections to the straight lines spanned by the vectors~$e_j$. It follows that the random variables~$\eta_k$ with law~$\ell$ can be written in the form~\eqref{2.3}, where~$\xi_{jk}$ is a normal random variable with variance~$1$, and therefore its law possesses an infinitely smooth density satisfying~\eqref{8.6} and~\eqref{10.15}. The validity of~\eqref{1.07} is implied  by the Fernique theorem (see Theorem~2.8.5 in~\cite{bogachev1998}).

\begin{proof}[Proof of Theorem~\ref{t2.2}]
We shall prove that the hypotheses of Theorem~\ref{t1.1} hold for the Markov family in question. This will imply all the required results. 

\medskip
{\it Step~1: Continuity and compactness\/}. 
We claim that Condition~(A) is satisfied for the pair~$(H,U)$, where~$H$ is defined by~\eqref{2.0} and $U=V^{s+1}$. To see this, we apply a standard regularisation property for the 2D Navier--Stokes equations. Namely, as is proved in Chapter~17 of~\cite{taylor1996} (see also Theorem~2.1.19 in~\cite{KS-book}), the time-$1$ shift $S:H\to H$ along trajectories of the deterministic Navier--Stokes system~\eqref{2.1} (in which $f(t)\equiv h\in V^s$) maps~$H$ to~$V^{s+2}$. Moreover, the image by~$S$ of any ball in~$H$ is a bounded subset in~$V^{s+2}$. Since $S:H\to H$ is continuous and the embedding $V^{s+2}\subset V^{s+1}$ is compact, it follows that the mapping $S:H\to V^{s+1}$ is continuous and maps any ball of~$H$ to a relatively compact subset. 

\smallskip
{\it Step~2: Dissipativity\/}.
We claim that inequality~\eqref{1.1} holds with $\varPhi(u)=\|u\|$, $q=e^{-\nu\alpha_1}$, and a sufficiently large~$C$. Indeed, it is well known that (e.g., see inequality~(2.25) in~\cite[Chapter~III]{temam1988})
$$
\|S(u)\|\le q\|u\|+C, \quad u\in H, 
$$
where $C\ge1$ does not depend on~$u$. It follows that
$$
\varPhi(S(u)+v)\le q\|u\|+C+\|v\|\le q\,\varPhi(u)+C\bigl(\varPhi(v)+1\bigr). 
$$

\smallskip
{\it Step~3: Structure of the noise\/}. 
The fact that $\supp\ell=H$ follows from~\eqref{2.3} and the positivity of the coefficients~$b_j$ and densities~$\tilde \rho_j$. The validity of~\eqref{1.07} is required by the hypotheses of the theorem. It remains to prove that $V^{s+1}\subset H_\ell$ (where~$H_\ell$ stands for the set of admissible shifts of~$\ell$; see Section~\ref{s8.1}) and that the mapping $\theta:V^{s+1}\to\PP(H)$ taking a vector~$a$ to the shifted measure~$\ell_a$ is continuous. To this end, we shall show that inequality~\eqref{10.13} holds, which implies that the hypotheses of Proposition~\ref{p5.3} are satisfied. Denoting by~$C_s$ the sum of the series in~\eqref{2.5} and using the Cauchy--Schwarz inequality, we derive
$$
\sum_{j=1}^\infty b_j^{-1}|(v,e_j)|
\le \biggl(\,\sum_{j=1}^\infty b_j^{-2}\alpha_j^{-1-s}\biggr)^{1/2}
\biggl(\,\sum_{j=1}^\infty |(v,e_j)|^2\alpha_j^{s+1}\biggr)^{1/2}
=C_s^{1/2}\|v\|_{V^{s+1}}. 
$$

We have thus shown that Hypotheses (A)--(C) are satisfied with $q=e^{-\nu\alpha_1}$ and any sufficiently large~$C>0$. This completes the proof of the theorem. 
\end{proof}

\begin{corollary} \label{c2.3}
In addition to the hypotheses of Theorem~\ref{t2.2}, assume that the law of~$\eta_k$ and the initial measure  $\lambda\in\PP(H)$ satisfy the conditions 
\begin{equation} \label{2.11}
\int_H\exp(\alpha \|u\|^2)\ell(\dd u)<\infty,\quad 
\int_H\exp(\alpha \|u\|^2)\lambda(\dd u)<\infty
\end{equation}
for some~$\alpha>0$.
Then, for any $\theta\in(0,2)$, the $\IP_\lambda$-laws of the random variables
$$
\xi_k=\frac1k\sum_{n=0}^{k-1}\|u_n\|^\theta
$$
satisfy the LDP with a good rate function $I:\R\to[0,+\infty]$ not depending on~$\lambda$. 
\end{corollary}

\begin{proof}
As was shown above, the hypotheses of Theorem~\ref{t1.1} are satisfied for the Markov family~$(u_k,\IP_u)$. Therefore, the required result will be established if we prove that the conditions of Theorem~\ref{t1.2} hold for some uniformly stabilisable functional~$\pppp$. 

\smallskip
For $\e>0$, let us set $\pppp_\e(u)=\e\|u\|^2$. We claim that if~$\ell$ satisfies the first inequality in~\eqref{2.11}, then
\begin{equation} \label{2.12}
\E_u\exp\bigl(\pppp_\e(u_1)+\cdots+\pppp_\e(u_k)\bigr)
\le \exp(C\e\|u\|^2+Ck), \quad k\ge1,
\end{equation}
where~$C>0$ is an absolute constant and~$\e>0$ is sufficiently small. Indeed, it is well known that (e.g., see inequality~(2.53) in~\cite{KS-book})
$$
\|v\|^2\le C_1\biggl(\int_0^1\|S_t(v)\|_1^2\dd t+1\biggr),\quad v\in H,
$$
where $C_1>0$ does not depend on~$v$, and~$S_t:H\to H$ stands for the time-$t$ shift along trajectories of Eq.~\eqref{2.1} with $f\equiv h$. It follows that 
$$
\pppp_\e(u_1)+\cdots+\pppp_\e(u_k)\le C_1\e\sum_{l=1}^k\int_0^1\|S_t(u_l)\|_1^2\dd t+C_1\e k. 
$$
As is proved in Step~2 of the proof of Proposition~2.3.8 in~\cite{KS-book}, the mean value of the exponential of the right-hand side of this inequality can be estimated by the right-hand side of~\eqref{2.12}. Thus, the functional~$\pppp_\e$  is uniformly stabilisable  and satisfies inequality~\eqref{1.9} with $Q(r)=\exp(C\e r^2)$. It remains to note that convergence~\eqref{1.11} holds for the continuous function $f(v)=\|v\|^\theta$, and condition~\eqref{1.17} is fulfilled for $\e\ll1$ and any measure~$\lambda\in\PP(H)$ satisfying the second inequality in~\eqref{2.11} with some $\alpha>0$. 
\end{proof}

\subsection{Complex Ginzburg--Landau equation}
\label{s9.2}
We consider the following equation on the torus $\T^d\subset\R^d$:
\begin{equation} \label{2.13}
\p_t u-(\nu+i)(\Delta-1)u+ia|u|^2u=f(t,x), \quad x\in \T^d.
\end{equation}
Here $a>0$ is a parameter, $u=u(t,x)$ is a complex-valued function, and~$f$ is a random process. We assume that~$f$ has the form~\eqref{2.2}, where $h\in L^2(\T^d)$ is a deterministic  complex-valued function and~$\{\eta_k\}$ is a sequence of i.i.d.\;random variables in the complex space~$H^1(\T^d)$, where~$H^s(\T^d)=:V^s$ is the Sobolev space of order~$s$. If $d\le 4$, then the Cauchy problem for~\eqref{2.13} is well posed in~$V^1$ (e.g., see~\cite{weissler-1980,GV-1996,cazenave2003}). This means that, for any $u_0\in V^1$, problem~\eqref{2.13} has a unique solution satisfying the initial condition
\begin{equation} \label{2.28}
u(0,x)=u_0(x). 
\end{equation}
Under the above hypotheses, the restrictions of solutions to~\eqref{2.13} form a discrete-time Markov process~$(u_k,\IP_u)$ in the space~$V^1$, which is regarded as a real Hilbert space with the scalar product 
$$
(u,v)_1=(u,v)+\sum_{j=1}^d(\p_j u,\p_j v), \quad 
(u,v)=\Re\int_{\T^d}u\bar v\,\dd x.
$$
Let~$\{e_j\}$ be the complete system of eigenfunctions of~$-\Delta+1$, which are indexed so that the corresponding  eigenvalues~$\{\alpha_j\}$ form a non-decreasing sequence. We normalise~$e_j$ to be unit vectors in~$V:=V^1$. In what follows, we impose the following condition on~$\eta_k$. 

\begin{condition} \label{c2.4}
The i.i.d.\;random variables~$\eta$ have the form~\eqref{2.3}, where~$\{b_j\}$ is a sequence of positive numbers satisfying~\eqref{2.4}, $\xi_{jk}=\xi_{jk}^1+i\xi_{jk}^2$, and~$\xi_{jk}^l$ are independent real-valued random variables. Moreover, the laws of~$\xi_{jk}^l$ possess densities~$\tilde\rho_j^l\in C^1$ with respect to the Lebesgue measure, which are positive and satisfy~\eqref{8.6} and~\eqref{10.15}.  
\end{condition}

Let us define the functional
$$
\HH(u)=\int_{\T^d}
\Bigl(\frac12|\nabla u(x)|^2+\frac12|u(x)|^2+\frac{a}{4}|u(x)|^4\Bigr)\,\dd x.
$$
The following result is an analogue of Theorem~\ref{t2.2} in the case of the Ginzburg--Landau equation. Its proof is essentially the same, and we shall confine ourselves to outlining it. 

\begin{theorem} \label{t2.5}
Let $s\ge d$ be an integer, let $h\in V^{s-1}$, and let~$\{\eta_k\}$ be a sequence  random variables for which Condition~\ref{c2.4} is fulfilled. Assume, in addition, that  the law~$\ell$ of~$\eta_k$ satisfies~\eqref{1.07} with $\varPhi(u)=(\HH(u))^\theta$ for some positive numbers~$\delta$ and~$\theta$, and inequality~\eqref{2.5} holds. Then $(u_k,\IP_u)$ has a unique stationary measure~$\mu\in\PP(V)$, which is exponentially mixing in the sense that~\eqref{1.7} holds with $H=V$. Moreover, for any $c>0$ and any subset $\Lambda\subset\PP(V)$ satisfying condition~\eqref{1.8} in which $H=V$, the uniform LDP with $\lambda\in\Lambda$ and a good rate function\footnote{We define $\VVV=V^{\Z_+}$.} $\III:\PP(\VVV)\to[0,+\infty]$ holds for the sequence of $\IP_\lambda$-occupation measures~\eqref{1.12}. 
\end{theorem}

\begin{proof}[Outline of the proof]
We need to check Hypotheses (A)--(C), in which $\varPhi(u)$ is defined in the statement of the theorem, and~$S:V\to V$ stands for the time-$1$ shift along trajectories of problem~\eqref{2.13} with $f(t)\equiv h$. The validity of~(A) with $U=V^s$ is a standard fact of the regularity theory for parabolic systems. Indeed, using Proposition~1.1 of~\cite[Chapter~15]{taylor1996}, one can prove the local existence, uniqueness, and regularity of a solution. To show that the solutions are global, it suffices to derive an a priori bound on the $H^1$ norm. This property is an immediate consequence of inequality~\eqref{2.30} established below. To check~(B), let us note that the Fr\'echet derivative of~$\HH(u)$ calculated on a vector~$v\in H_0^1$ has the form
$$
\HH'(u;v)=\Re\int_{\T^d}\bigl(\nabla u\cdot\nabla\bar v+(1+a|u|^2)u\bar v\bigr)\,\dd x.
$$
It follows that if~$u=u(t,x)$ is a solution of~\eqref{2.13}, then
\begin{align*}
\frac{\dd}{\dd t}\HH(u)
&=\bigl((1-\Delta)u+a|u|^2u,-(\nu+i)(1-\Delta)u-ia|u|^2u+f\bigr)\\
&\le-\nu\bigl(\|(1-\Delta)u\|^2+a(|u|^2,|\nabla u|^2)+a\|u\|_{L^4}^4\bigr)
+\bigl((1-\Delta)u+a|u|^2u,f\bigr),
\end{align*}
where we used the relations 
$$
(v,iv)=0,\quad (|u|^2u,\Delta u)\le (|u|^2,|\nabla u|^2).
$$
Taking $f(t)\equiv h$ and applying the Friedrichs and Cauchy--Schwarz inequalities, we derive
$$
\frac{\dd}{\dd t}\HH(u(t))\le -\beta\HH(u(t))+M,
$$
where $M=C(\|h\|_{L^4}^4+1)$ and $\beta>0$. The Gronwall inequality now implies that
\begin{equation} \label{2.30}
\HH(S(u))\le e^{-\beta}\HH(u)+\beta^{-1}M.
\end{equation}
It is easy to see that $\HH(z+v)\le (1+\alpha)\HH(z)+C_\alpha\HH(v)$ for any $u,v\in V$, where $\alpha>0$ is arbitrary and~$C_\alpha>0$ depends only on~$\alpha$. Combining this inequality with~\eqref{2.30}, we obtain
$$
\HH(S(u)+v)\le (1+\alpha)e^{-\beta}\HH(u)+C_\alpha\HH(v)+(1+\alpha)\beta^{-1}M.
$$
Choosing~$\alpha>0$ sufficiently small and raising the resulting inequality to power $\theta>0$, we arrive at~\eqref{1.1} with $\varPhi(u)=(\HH(u))^\theta$ and $H=V$. 

Finally, let us show that~(C) holds. The fact that the support of~$\ell$ coincides with~$V$ follows from the positivity of the coefficients~$b_j$ and of the densities for the one-dimensional projections of~$\ell$. Inequality~\eqref{1.07} is required to hold by hypothesis. Thus, it remains to check that $V^s\subset H_\ell$ and that the mapping $\theta:V^s\to\PP(V)$ taking~$a$ to~$\ell_a$ is continuous. By Proposition~\ref{p5.3}, these properties will be established if we prove that inequality~\eqref{10.13} holds with $U=V^s$ and the orthonormal basis of~$V$ formed of the vectors $\{e_j,ie_j,j\ge1\}$, where~$e_j$ are the $V$-normalised eigenfunctions of the Laplacian on~$\T^d$. To prove~\eqref{10.13}, it suffices to note that, in view of~\eqref{2.5}, we have
\begin{align*}
\biggl(\,\sum_{j=1}^\infty b_j^{-1}\bigl(|(v,e_j)|+|(v,ie_j)|\bigr)\biggr)^2
&\le\sum_{j=1}^\infty b_j^{-2}\alpha_j^{-1-s}\,\sum_{j=1}^\infty\alpha_j^{s}\bigl(|(v,\hat e_j)|+|(v,i\hat e_j)|\bigr)^2\\
&\le C\,\|v\|_{V^{s}}^2,
\end{align*}
where $e_j=\sqrt{\alpha_j}\,\hat e_j$. This completes the proof of Theorem~\ref{t2.5}. 
\end{proof}

As in the case of the Navier--Stokes system, we can derive from Theorem~\ref{t2.5} some results on LDP for observables with moderate growth at infinity. To simplify the presentation, we shall consider only the case~$\theta=\frac12$, which covers Gaussian perturbations. 

\begin{corollary} \label{c2.6} 
In addition to the hypotheses of Theorem~\ref{t2.5}, assume that the law~$\ell$  of~$\eta_k$ and the initial measure~$\lambda\in\PP(V)$ satisfy the conditions
\begin{equation} \label{2.31}
\int_V\exp\bigl(\alpha\sqrt{\HH(u)}\bigr)\ell(\dd u)<\infty, \quad 
\int_V\exp\bigl(\alpha\sqrt{\HH(u)}\bigr)\lambda(\dd u)<\infty,
\end{equation}
where $\alpha>0$. Then, for any measurable function $f:V\to\R$ satisfying the condition $\frac{|f(u)|}{\sqrt{\HH(u)}}\to0$ as $\|u\|_V\to\infty$, 
the $\IP_\lambda$-laws of the random variables
$$
\xi_k=\frac1k\sum_{n=0}^{k-1}f(u_k)
$$
satisfy the LDP with a good rate function not depending on~$\lambda$. 
\end{corollary}

\begin{proof}
As for the proof of Corollary~\ref{c2.3}, it suffices to show that $\pppp_\e(u)=\e\sqrt{\HH(u)}$  is a uniformly stabilisable functional. To this end, we use inequality~\eqref{1.1} with $\varPhi(u)=\sqrt{\HH(u)}$. Setting $u=u_{n-1}$ and $v=\eta_n$ with $n=1,\dots, k$, we derive
\begin{equation} \label{2.013}
\varPhi(u_n)\le q\,\varPhi(u_{n-1})+C(\varPhi(\eta_n)+1). 
\end{equation}
Summing up these inequalities, we obtain
$$
\sum_{n=1}^k\varPhi(u_n)\le C_1\varPhi(u)+C_1\sum_{n=1}^k\varPhi(\eta_n)+C_1k. 
$$
The independence of~$\eta_k$ now implies that
\begin{equation} \label{2.014}
\E_u\exp(\pppp_\e(u_1)+\cdots+\pppp_\e(u_k))
\le e^{\e C_1(\varPhi(u)+k)}\biggl(\int_Ve^{\e C_1\varPhi(z)}\ell(\dd z)\biggr)^k. 
\end{equation}
Taking into account the first condition in~\eqref{2.31}, we see that~$\pppp_\e$ is  uniformly stabilisable for~$\e\ll1$. It remains to note that, in view of the second condition in~\eqref{2.31}, inequality~\eqref{1.17} is also satisfied for~$\e\ll1$. 
\end{proof}

\subsection{Burgers equation}
\label{s9.3}
Let us consider the problem~\eqref{0.3}--\eqref{0.5}. Our aim is to establish Theorems~A and~B stated in the Introduction. In view of Theorem~\ref{t1.1}, to prove Theorem~A, it suffices to check the validity of Hypotheses~(A)--(C), in which $U=V^{s+1}$. The fact that $S:H\to V^{s+1}$ is continuous and bounded on any ball is a standard regularity result, and we omit it. Inequality~\eqref{1.1} with $\varPhi(u)=\|u\|$ is also well known, and the validity of~\eqref{1.07} with any $\delta>0$ follows from the Fernique theorem; e.g., see Theorem~2.8.5 in~\cite{bogachev1998}. To check the remaining hypotheses in~(C), recall that the semigroup of admissible shifts for a Gaussian measure coincides with its Cameron--Martin space; see Theorem~2.4.5 in~\cite{bogachev1998}. Hence, the continuous inclusion of~$U=V^{s+1}$ into~$H_\ell$ holds in view of the hypotheses of Theorem~A. Finally, to prove the continuity of $\theta:V^{s+1}\to\PP(H)$, we use the following estimate for the total variation norm between shifts of a Gaussian measure (see Lemma~2.4.4 in~\cite{bogachev1998}): 
\begin{equation} \label{2.32}
\|\ell_{a}-\ell_{a'}\|_{\mathrm{var}}\le 
2\bigl(1-\exp\bigl\{-\tfrac14\|a-a'\|_{H_\ell}^2\bigr\}\bigr)^{1/2}.
\end{equation}
Here $a,a'\in H_\ell$ are arbitrary vectors, and~$\|\cdot\|_{H_\ell}$ denotes the norm in the Cameron--Martin space of~$\ell$:
$$
\|a\|_{H_\ell}^2=\sum_{j=1}^\infty b_j^{-2}a_{j}^2, \quad a=(a_1,a_2,\dots),
$$
where $a$ is expanded in the eigenbasis of the covariance operator for~$\ell$. Since~$V^{s+1}$ is continuously embedded in~$H_\ell$, we see that the shift operator~$\theta$ is continuous from~$V^{s+1}$ to~$\PP(H)$. This completes the proof of Theorem~A. 

\smallskip
We now turn to Theorem~B. In view of Theorem~\ref{t1.5}, to prove the LDP and the Gallavotti--Cohen relation for the rate function, it suffices to find a uniformly stabilisable function $\pppp:H\to\R_+$ such that~\eqref{1.13} holds and to check~\eqref{1.17} for $\lambda=\ell$. Exactly the same argument as for the 2D Navier--Stokes system or the Ginzburg--Landau equation shows that $\pppp_\e(u)=\e\|u\|^2$ with $\e>0$ is a uniformly stabilisable functional, and the corresponding function~$Q$ entering~\eqref{1.9} can be chosen to be $Q_\e(r)=\exp(C\e r^2)$, where $C>0$ does not depend on~$\e$. By Fernique's theorem, it follows that condition~\eqref{1.17} is satisfied for~$\ell$, provided that~$\e>0$ is sufficiently small. 

We now prove the boundedness of~$\sigma(v_0,v_1)$ on balls of~$H\times H$ and the convergence relation~\eqref{1.13}. By the hypotheses of Theorem~B, the measure~$\ell$ can be decomposed in  the standard trigonometric basis in~$H$ and written in the form~\eqref{8.1}, where~$\mu_j$ denotes the centred normal law on~$\R$ with variance~$b_j^2$. It follows from~\eqref{8.4} that
\begin{equation} \label{2.33}
\rho(u,v)=\exp\bigl(-\tfrac12\|S(u)\|_b^2+(S(u),v)_b\bigr),
\end{equation}
where we set 
$$
(u,v)_b=\sum_{j=1}^\infty b_j^{-2}u_jv_j, \quad \|u\|_b=(u,u)_b^{1/2}. 
$$
Combining~\eqref{2.33} and~\eqref{entropy}, we see that  
\begin{equation} \label{2.35}
\sigma(u,v)=\tfrac12\|S(v)\|_b^2-\tfrac12\|S(u)\|_b^2+(S(u),v)_b-(S(v),u)_b. 
\end{equation}
We now need the following lemma, which is a consequence of the Kruzhkov maximum principle~\cite{kruzhkov-1969}; its proof in the more difficult stochastic case can be found in~\cite[Section~3]{boritchev-2013}.

\begin{lemma} \label{l2.7}
Let $h\in V^m$ for some integer $m\ge2$. Then the image of~$S$ is contained in~$V^{m+1}$, the mapping $S:H\to V^{m+1}$ is continuous, and there is $K_m>0$ such that
\begin{equation} \label{2.34}
\|S(u)\|_{m+1}\le K_m\quad\mbox{for any $u\in H$}. 
\end{equation}
\end{lemma}

Now note that the continuity of the embedding $V^{s+1}\subset H_\ell$ implies the inequality
\begin{equation} \label{2.61}
\|w\|_{b^2}^2:=\sum_{j=1}^\infty |w_j|^2b_j^{-4}
\le C\sum_{j=1}^\infty |w_j|^2(1+|j|^2)^{2(s+1)}=C\|w\|_{2(s+1)}^2, 
\end{equation}
where $w\in V^{2(s+1)}$ and  $w_j=(w,e_j)$. 
Combining this with inequality~\eqref{2.34} and relation~\eqref{2.35}, we obtain
\begin{align}
&|\sigma(u,v)|
\le \frac12\bigl(\|S(u)\|_b^2+\|S(v)\|_b^2\bigr)+\|u\|\,\|S(v)\|_{b^2}
+\|u\|\,\|S(v)\|_{b^2}\notag\\
&\quad\le \frac C2\bigl(\|S(u)\|_{s+1}^2+\|S(v)\|_{s+1}^2\bigr)
+C\|u\|\,\|S(v)\|_{2(s+1)}+C\|u\|\,\|S(v)\|_{2(s+1)}\notag\\
&\quad\le CK_s^2+CK_{2s+1}\bigl(\|u\|+\|v\|\bigr). 
\label{2.65}
\end{align}
We see that Condition~(D) is fulfilled for the Burgers equation. Thus, the LDP and the Gallavotti--Cohen symmetry hold for the entropy production and the corresponding rate function.  We note that a similar argument combined with~\eqref{2.33} shows that 
$$
e^{-C(1+\|v\|)}\le \rho(u,v)\le e^{C(1+\|v\|)}\quad\mbox{for any $u,v\in H$}.
$$
Integrating with respect to a measure~$\lambda\in\PP(H)$, we derive the following rough estimate on the density of the measure~$\mathfrak{P}^\ast\lambda$:
$$
\mathrm{e}^{-C(1+\|v\|)}\le \frac{\dd\mathfrak{P}^\ast\lambda}{\dd\ell}(v)\le\mathrm{e}^{C(1+\|v\|)},
$$
where the constant $C$ is independent of~$\lambda$. In particular, this is true for the stationary measure $\mu=\mathfrak{P}^\ast\mu$. Moreover, the above estimate implies that the boundary term in the entropy balance relation~\eqref{kBal} is indeed~$O(k^{-1})$ for large~$k$, provided that the measure~$\lambda$ satisfies~\eqref{1.17}.

\smallskip
It remains to prove the positivity and finiteness of the mean entropy production~$\langle\sigma\rangle_\mu$. As was explained in the introduction, we always have $\langle\sigma\rangle_\mu\ge0$, and the equality holds if and only if
the detailed balance condition~\eqref{DBCond} is satisfied $\ell\otimes\ell$ 
almost everywhere. Recalling~\eqref{2.33}, we can write this condition as
\begin{multline} \label{2.36}
\exp\bigl(-\tfrac12\|S(v)\|_b^2+(S(v),u)_b\bigr)
\int_H\rho(z,v)\mu(\dd z)\\
=\exp\bigl(-\tfrac12\|S(u)\|_b^2+(S(u),v)_b\bigr)
\int_H\rho(z,u)\mu(\dd z).
\end{multline}
It follows from inequality~\eqref{2.61} and Lemma~\ref{l2.7} that the expressions under the exponents are continuous functions on~$H\times H$. Moreover, the function~$\rho(z,v)$ is also continuous on~$H\times H$ and is bounded by $e^{C\|v\|}$ uniformly in~$z$. Applying the dominated convergence theorem, we see that the integrals in~\eqref{2.36} are also continuous functions. Since~$\supp(\ell\otimes\ell)$ coincides with the whole space, we see that relation~\eqref{2.36} must hold for all $(u,v)\in H\times H$.  Taking the logarithm of both sides of~\eqref{2.36}, replacing~$v$ by $\lambda v$, and dividing by~$\lambda$, we derive
\begin{equation} \label{2.37}
(S(u),v)_b=\frac{1}{\lambda}\log
\int_He^{\lambda(v,S(z))_b}\exp\bigl(-\tfrac12\|S(z)\|_b^2\bigr)\mu(\dd z)
+\lambda^{-1}r(\lambda),
\end{equation}
where we set
$$
r(\lambda)=\tfrac12\bigl(\|S(u)\|_b^2-\|S(\lambda v)\|_b^2\bigr)
+(S(\lambda v),u)_b
+\log\int_H\rho(z,u)\mu(\dd z). 
$$
It follows from inequality~\eqref{2.34} with $m=2s+1$ that~$r$ is a bounded function of~$\lambda\in\R$, so that the second term on the right-hand side of~\eqref{2.37} goes to zero as $\lambda\to+\infty$. Since the first term on the right-hand side does not depend on~$u$, passing to the limit in~\eqref{2.37} as~$\lambda\to+\infty$, we conclude that
$$
\bigl(S(u),v\bigr)_b=C(v)
\quad\mbox{for all $u,v\in H$},
$$
where $C(v)$ depends only on~$v$. 
It follows that $S(v)$ is a constant function on~$H$. This contradicts the backward uniqueness of solutions for the Burgers equation; e.g., see Section~II.8 in~\cite{BV1992} for the more complicated case of quasilinear parabolic equations. 

To prove the finiteness of~$\langle\sigma\rangle_\mu$, note that, in view of~\eqref{2.65}, we have  
\begin{equation} \label{2.38}
\langle\sigma\rangle_\mu\le 
\int_{H\times H}|\sigma(u,v)|\,\mmu(\dd u,\dd v)
\le C\biggl(1+\int_H\|z\|\mu(\dd z)\biggr). 
\end{equation}
The integral on the right-hand side of this inequality is equal to
$$
\E_\mu \|u_1\|\le \E_\mu \|S(u_0)\|+\E\|\eta_1\|<\infty,
$$
where we used inequality~\eqref{2.34}. The proof  of Theorem~B is complete. 

\subsection{Reaction-diffusion system}
\label{s9.4}
Let $D\subset\R^d$ be a bounded domain with $C^\infty$-smooth boundary~$\p D$. We consider the problem
\begin{align}
\dot u-a\Delta u+g(u)&=f(t,x), \label{2.51}\\
u\bigr|_{\p D}&=0,\label{2.52}\\
u(0,x)&=u_0(x). \label{2.53}
\end{align}
Here $u=(u_1,\dots,u_l)^t$ is an unknown vector function, $a$ is an $l\times l$ matrix such that 
\begin{equation} \label{2.54}
a+a^t>0,
\end{equation}
$g\in C^\infty(\R^l,\R^l)$ is a given function, and~$f$ is a random process of the form~\eqref{2.2}. We assume that~$g$ satisfies the following growth and dissipativity conditions: 
\begin{align}
\langle g(u),u\rangle&\ge -C+c|u|^{p+1},\label{2.55}\\ 
g'(u)+g'(u)^t&\ge -C I,\label{2.56}\\
|g'(u)|&\le C(1+|u|)^{p-1},\label{2.57}
\end{align}
where $\langle \cdot,\cdot\rangle$ stands for the scalar product in~$\R^l$, $g'(u)$ is the Jacobi matrix for~$g$, $I$~is the identity matrix, $c$~and~$C$ are positive constants, and $1< p\le \frac{d+2}{d-2}$. As in the case of the 2D Navier--Stokes system, problem~\eqref{2.51}--\eqref{2.53} is well posed (e.g., see Sections~1.4 and~1.5 in~\cite{BV1992}) and generates a discrete-time Markov process denoted by~$(u_k,\IP_u)$. Our aim is to study the LDP for the occupation measures~\eqref{1.12}. 

Let us denote by~$\{e_j\}$ an orthonormal basis in~$H=L^2(D,\R^l)$ composed of the eigenfunctions of the Dirichlet Laplacian~$-\Delta$ and by~$V^s$ the domain of the operator $(-\Delta)^{s/2}$. 

\begin{theorem} \label{t2.8}
In addition to the above hypotheses, assume that $s\ge d$ is an integer, $h\in V^s$, the function~$g(u)$ belongs to~$C^s$ and vanishes at $u=0$ together with its derivatives up to order~$s$, and~$\{\eta_k\}$ is an 
i.i.d.\;sequence of random variables satisfying Condition~\ref{c2.1} such that~\eqref{2.5} and~\eqref{1.07} hold with $\varPhi(u)=\|u\|$. Then $(u_k,\IP_u)$ has a unique stationary measure~$\mu\in\PP(H)$, which is exponentially mixing. Moreover, for any $c>0$ and any subset $\Lambda\subset\PP(H)$ satisfying~\eqref{1.8} the uniform LDP with $\lambda\in\Lambda$ and a good rate function $\III:\PP(\HHH)\to[0,+\infty]$ holds for the sequence of $\IP_\lambda$-occupation measures~\eqref{1.12}. 
\end{theorem}

This theorem can be established by a literal repetition of the arguments used in the case of the Navier--Stokes system. The only difference is that the equation is considered on a bounded domain, and to have regularising property for solutions, we need to impose some additional hypotheses. This is the reason for requiring~$h$ to be in the domain of~$(-\Delta)^{s/2}$ and~$g$ to vanish at zero together with its derivatives up to order~$s$. Since the corresponding arguments are standard, we omit the proof of Theorem~\ref{t2.8}. 

We now turn to the Gallavotti--Cohen fluctuation principle. The following result is an analogue of Theorem~B for the reaction--diffusion system.

\begin{theorem} \label{t2.9}
In addition to the hypotheses of Theorem~\ref{t2.8}, let us assume that $h\in V^{2s+1}$, the function~$g$ belongs to~$C^{2s+1}$ and vanishes at $u=0$ together with its derivatives up to order~$2s+1$, the orthonormal basis entering the decomposition~\eqref{8.1} for the measure~$\ell$ coincides with the eigenbasis~$\{e_j\}$, and the measure~$\ell$ satisfies the first condition in~\eqref{2.11}. Furthermore, suppose that for any $A_0>0$ there is positive number $C_1=C_1(A_0)$ such that the second inequality in~\eqref{10.15} holds for $y\in\R$ and $A\in[0,A_0]$. Then, for any initial point $u_0\in H$, the laws of the random variables~\eqref{1.14}, in which $\sigma(v_0,v_1)$ is the entropy production functional for~$(u_k,\IP_u)$, satisfy the LDP with a good rate function~$I:\R\to[0,+\infty]$ not depending on~$u_0$. Moreover, the Gallavotti--Cohen fluctuation relation~\eqref{1.15} holds for~$I$. 
\end{theorem}

Before proving this result, let us check that the conditions imposed on~$\ell$ are satisfied for any centred Gaussian measure on~$H$ such that its Cameron--Martin space contains~$V^{r}$ with some $r<s+1-d/2$, and the eigenvectors of its covariance operator coincide with the eigenbasis~$\{e_j\}$ of the Dirichlet Laplacian in~$L^2(D,\R^l)$. Indeed, it is well known that a centred Gaussian measure~$\mu=\ell$ is representable in the form~\eqref{8.1}, where~$\mu_j$ denotes the projection of~$\mu$ to the straight line spanned by the $j^{\text{th}}$ eigenvector of the covariance operator of~$\mu$. It follows that, if~$\{\eta_k\}$ is a sequence of 
i.i.d.\;random variables in~$H$ with law~$\mu$, then~$\eta_k$ can be written in the form~\eqref{2.3}, where~$\xi_{jk}$ has a normal law (with zero mean value and variance~$1$) and~$b_j^2$ is the $j^{\text{th}}$ eigenvalue of the covariance operator of~$\mu$. In particular, we have
$$
\tilde \rho_j(r)=\frac{1}{\sqrt{2\pi}}\,e^{-r^2/2}, \quad r\in\R,
$$
whence it follows that~\eqref{8.6} and~\eqref{10.15} are satisfied with $C_1=A_0+2$. The validity of the first inequality in~\eqref{2.11} (which implies also~\eqref{1.07}) is a consequence of Fernique's theorem; see Theorem~2.8.5 in~\cite{bogachev1998}. Finally, to establish~\eqref{2.5}, note that the inclusion $V^r\subset H_\ell$ and the closed graph theorem imply the inequality
$$
\|u\|_{H_\ell}^2=\sum_{j=1}^\infty b_j^{-2}(u,e_j)^2\le C\sum_{j=1}^\infty \alpha_j^r(u,e_j)^2=C\|u\|_{V^r}^2 \quad
\mbox{for $u\in V^r$},
$$
whence it follows that $b_j^{-2}\le C\alpha_j^r$ for all $j\ge1$. Combining this with the asymptotic relation $\alpha_j\sim j^{2/d}$ as $j\to\infty$ (see Section~8.3 in~\cite{taylor1996}), we conclude that~\eqref{2.5} holds if $r<s+1-d/2$. 

\begin{proof}[Proof of Theorem~\ref{t2.9}]
We shall show that the hypotheses of Theorem~\ref{t1.5} are fulfilled; this will imply all required results. The verification of Conditions~(A) and~(B), in which $U=V^{2(s+1)}$ and $\varPhi(u)=\|u\|$, is similar to the case of the Navier--Stokes system, and therefore we only sketch it. 

The regularising property and boundedness of~$S$ are discussed below (see Lemma~\ref{l2.10}). The continuity of $S:H\to U$ follows from the continuity of~$S$ as a mapping in~$H$ and the compactness in~$U$ of the image of any ball. To establish the dissipativity, note that inequality~\eqref{2.037} established below implies that
$$
\p_t\|u\|^2+\tfrac{\delta\alpha_1}{2}\|u\|^2\le C(1+\|h\|^2). 
$$
Applying the Gronwall inequality, we easily prove~\eqref{1.1}. 

\smallskip
We now check Condition~(D). To prove the positivity of~$\rho$ and the continuity of the shift operator $\theta:V^{2(s+1)}\to\PP(H)$, in view of Proposition~\ref{p5.3}, it suffices to check inequality~\eqref{10.18}.  To this end, we first note that~\eqref{2.5} implies the inequality
\begin{equation} \label{2.60}
\sum_{j=1}^\infty b_j^{-2}|w_j|^2
\le\sup_{j\ge1}\bigl(|w_j|^2\alpha_j^{s+1}\bigr)\sum_{j=1}^\infty b_j^{-2}\alpha_j^{-s-1}\le C_1\|w\|_{V^{s+1}}^2, 
\end{equation}
where $w\in V^{s+1}$ and $w_j=(w,e_j)$. Setting $w_j=v_j\alpha_j^{\frac{s+1}{2}}$ in~\eqref{2.60} and using again~\eqref{2.5}, we obtain
$$
\sum_{j=1}^\infty b_j^{-2}|v_j|
\le \biggl(\sum_{j=1}^\infty b_j^{-2}\alpha_j^{-s-1}\biggr)^{1/2}
\biggl(\sum_{j=1}^\infty b_j^{-2}|v_j|^2\alpha_j^{s+1}\biggr)^{1/2}\le C_2\|v\|_{V^{2(s+1)}}. 
$$

We now prove that $\sigma(v_0,v_1)$ is bounded on balls of~$H\times H$ and satisfies~\eqref{1.13} for some uniformly stabilisable functional~$\pppp$. To this end, we first repeat the argument used in the proof of Corollary~\ref{c2.6} to show that the functional $\pppp_\e(u)=\e\|u\|^2$ is uniformly stabilisable. Namely, inequality~\eqref{2.013} with $\varPhi(u)=\|u\|^2$ is valid for the reaction-diffusion equation. It follows that~\eqref{2.014} is also true. Recalling the first inequality in~\eqref{2.11}, we obtain~\eqref{1.9} with $Q(r)=\exp(C\e r^2)$. 

For $v_0,v_1\in H$, let us write
$$
v_i=\sum_{j=1}^\infty v_{ij}e_j, \quad S(v_i)=\sum_{j=1}^\infty S_j(v_i)e_j,
\quad i=0,1. 
$$
Let~$\rho_j$ be the density of the law of~$b_j\xi_{jk}$, $k\ge1$, so that $\rho_j(r)=b_j^{-1}\tilde \rho_j(r/b_j)$. Combining this relation with~\eqref{8.4} and~\eqref{entropy}, we obtain
\begin{align}
\sigma(v_0,v_1)
&=\sum_{j=1}^\infty\biggl(\log\frac{\rho_j\bigl(v_{1j}-S_j(v_0)\bigr)}{\rho_j(v_{1j})}-
\log\frac{\rho_j\bigl(v_{0j}-S_j(v_1)\bigr)}{\rho_j(v_{0j})}\biggr)
\notag\\
&=\sum_{j=1}^\infty \bigl(\Xi_j(v_0,v_1)-\Xi_j(v_1,v_0)\bigr), 
\label{2.58}
\end{align}
where we set
$$
\Xi_j(v_0,v_1)=\log\tilde\rho_j\bigl(\tfrac{v_{1j}-S_j(v_0)}{b_j}\bigr)
-\log\tilde\rho_j\bigl(\tfrac{v_{1j}}{b_j}\bigr).
$$
Let us define $A_{ij}=b_j^{-1}|S_j(v_i)|$ and $A=\sup_{i,j}A_{ij}$, where the supremum is taken over $i=0,1$ and $j\ge1$. The second inequality in~\eqref{10.15} implies that
\begin{align*}
|\Xi_j(v_0,v_1)|
&=\biggl|\int_0^1\frac{\dd}{\dd\theta}
\log\tilde\rho_j\biggl(\frac{v_{1j}-\theta S_j(v_0)}{b_j}\biggr)\dd\theta\biggr|
\le \int_{-A_{0j}}^{A_{0j}}
\frac{|\tilde\rho_j'(b_j^{-1}v_{1j}-r)|}{\tilde\rho_j(b_j^{-1}v_{1j}-r)}\dd r\\
&\le C_3(A)\bigl(b_j^{-1}|v_{1j}|+1\bigr)A_{0j}. 
\end{align*}
A similar inequality holds for~$\Xi(v_1,v_0)$.
Substituting these estimates into~\eqref{2.58} and using the Cauchy--Schwarz inequality and condition~\eqref{2.4}, we obtain 
\begin{align}
|\sigma(v_0,v_1)|
&\le C_4\sum_{j=1}^\infty
b_j^{-1}\Bigl(\bigl(b_j^{-1}|v_{1j}|+1\bigr)|S_j(v_0)|+\bigl(b_j^{-1}|v_{0j}|+1\bigr)|S_j(v_1)|\Bigr)\notag\\
&\le C_5\bigl(\|v_1\|+1\bigr)\|S(v_0)\|_{b^2}+C_5\bigl(\|v_0\|+1\bigr)\|S(v_1)\|_{b^2},\label{2.62}
\end{align}
where the norm $\|\cdot\|_{b^2}$ is defined in~\eqref{2.61}. We now need the following lemma, established at the end of this section. 

\begin{lemma} \label{l2.10}
Let \eqref{2.54}--\eqref{2.57} be satisfied, let $m\ge0$ be an integer, let $g\in C^m$ be a function vanishing at $u=0$ together with its derivatives up to order~$m$, and let $h\in V^m$. Then the image of~$S$ is contained in~$V^{m+1}$, and there is $K_m>0$ such that 
\begin{equation} \label{2.59}
\|S(v)\|_{V^{m+1}}\le K_m \quad\mbox{for any $v\in H$}. 
\end{equation}
\end{lemma}

It follows from~\eqref{2.60} and~\eqref{2.59} with $m=2s+1$ that
\begin{equation} \label{2.035}
\|S(v)\|_{b^2}^2\le C_6\sum_{j=1}^\infty \alpha_j^{2(s+1)}|S_j(v)|^2
=C_6\|S(v)\|_{2(s+1)}^2\le C_6K_{2s+1}^2, 
\end{equation}
where $v\in H$. 
Substituting this inequality into~\eqref{2.62}, we obtain 
\begin{equation} \label{2.70}
|\sigma(v_0,v_1)|\le C_7\bigl(\|v_0\|+\|v_1\|\bigr), \quad v_0,v_1\in H.
\end{equation}
This implies the required properties of~$\sigma$. 

\smallskip
It remains to show that~\eqref{1.17} is also satisfied. To this end, we note that the integrand in~\eqref{1.17} does not exceed $C_8\exp(C_8\e\|u\|^2)$. In view of the first inequality in~\eqref{2.11}, this function is integrable with respect to~$\ell$, provided that $\e>0$ is sufficiently small. The proof of Theorem~\ref{t2.9} is complete. 
\end{proof}

\begin{proof}[Proof of Lemma~\ref{l2.10}]
In view of the regularising property of the resolving operator for the reaction-diffusion system (see Proposition~7.7 in Section~15.7 of~\cite{taylor1996}), it suffices to prove that, if $u(t,x)$ is the solution of \eqref{2.51}--\eqref{2.53} with $f(t,x)\equiv h(x)$ and $u_0\in H$, then 
\begin{equation} \label{2.63}
\|u(\tfrac12,\cdot)\|\le K\quad\mbox{for any initial function $u_0\in H$},
\end{equation}
where $K>0$ does not depend on~$u_0$. 

Taking the scalar product in~$L^2$ of Eq.~\eqref{2.51} (in which $f\equiv h$) with~$2u$, we derive
$$
\p_t\|u\|^2+\int_D\bigl\langle(a+a^t)\nabla u,\nabla u\bigr\rangle\dd x+2\int_D\bigl(g(u),u\bigr)\dd x=2\int_D(h,u)\dd x.
$$
Using inequalities~\eqref{2.54} and~\eqref{2.55} to estimate the second and third terms on the left-hand side, we obtain
\begin{equation} \label{2.037}
\p_t\|u\|^2+\delta\|\nabla u\|^2+2c\|u\|_{L^{p+1}}^{p+1}\le C_1+\|h\|\,\|u\|,
\end{equation}
where $\delta$ and~$c$ are positive numbers. Since $\|u\|_{L^{p+1}}\ge C\|u\|$, we see that the function $\varphi(t)=\|u(t)\|^2$ satisfies the differential inequality
$$
\varphi'+2c_1\varphi^{(p+1)/2}\le C_2(1+\|h\|^2), 
$$
where $c_1>0$. It follows that, as long as $\varphi(t)\ge K_1:=\bigl(c_1^{-1}C_2(1+\|h\|^2)\bigr)^{2/(p+1)}$, we have
$$
\varphi'+c_1\varphi^{(p+1)/2}\le 0. 
$$
Recalling that $p>1$ and resolving this differential inequality, we obtain
$$
\varphi(t)\le \bigl(\varphi(s)^{(1-p)/2}+C_3(t-s)\bigr)^{-2/(p-1)}, \quad t\ge s\ge0.
$$
It follows that if $\varphi(0)\le K_1$, then $\varphi(t)\le K_1$ for all $t\ge0$, while if $\varphi(0)>K_1$, then
\begin{equation} \label{2.64}
\varphi(t)\le \bigl(\varphi(0)^{(1-p)/2}+C_3t\bigr)^{-2/(p-1)}\quad \mbox{for $0\le t\le T$}, 
\end{equation}
where $T>0$ is the first instant~$t>0$ such that $\varphi(t)=K_1$. Denoting by~$K_2$ the value of the right-hand side of~\eqref{2.64} with $t=1/2$ and $\varphi(0)=0$, we conclude that inequality~\eqref{2.63} holds with $K=\max(K_1^{1/2},K_2^{1/2})$.  
\end{proof}

Finally, we discuss briefly the question of strict positivity and finiteness of the entropy production rate. If the measure~$\ell$ is Gaussian, then exactly the same argument as in the case of the Burgers equation shows that the entropy production rate is strictly positive and finite in the stationary regime. However, these two properties are not related to the Gaussian structure of the noise and remain valid under more general hypotheses. Indeed, the finiteness of~$\langle\sigma\rangle_\mu$ follows from inequality~\eqref{2.70} and the fact the first moment of the stationary measure~$\mu$ is finite. On the other hand, the strict positivity of~$\langle\sigma\rangle_\mu$ holds under some additional hypotheses. Since the corresponding argument is technically rather complicated, we first outline the proof in the model case when $H=\R$. Namely, let us consider the Markov family associated with Eq.~\eqref{0.1} in which $S:\R\to\R$ is a non-constant continuous mapping with bounded image and~$\{\eta_k\}$ is a sequence of i.i.d.\;random variables in~$\R$ whose law~$\ell$ has a continuous density~$\theta$ against the Lebesgue measure that has the form
\begin{equation} \label{2.66}
\theta(y)=\exp\bigl(-a|y|^\beta+q(y)\bigr),\quad r\in\R,
\end{equation}
where $\beta\in(1,2]$, $a>0$, and~$q$ is a bounded continuous function.
As was explained in Section~\ref{s9.3}, the entropy production rate is zero if and only if (cf.~\eqref{2.36})
\begin{multline*} \label{2.68}
\exp\bigl(-a|v-S(u)|^\beta+q(v-S(u))\bigr)
\int_\R\rho(z,u)\mu(\dd z)\\
=\exp\bigl(-a|u-S(v)|^\beta+q(u-S(v))\bigr)
\int_\R\rho(z,v)\mu(\dd z),
\end{multline*}
where $\mu$ stands for the stationary distribution. 
Taking the logarithm of both sides of this relation and carrying out some simple transformations, we derive
\begin{equation} \label{2.67}
-a |v-S(u)|^\beta=\log\int_\R
\exp\bigl(-a|v-S(z)|^\beta+q(v-S(z))\bigr)\mu(\dd z)+r_1(u,v),
\end{equation}
where we denote by $r_i(u,v)$ some functions that are bounded in~$v$ for any fixed~$u$. Now note that
\begin{equation} \label{2.68}
|v-S(u)|^\beta=v^\beta-\beta S(u)v^{\beta-1}+r_2(u,v)v^{\beta-2}
\quad\mbox{as $v\to+\infty$}.
\end{equation}
Substituting this expression into~\eqref{2.67} and dividing by~$a\beta v^{\beta-1}$, we obtain
$$
S(u)=\frac{1}{a\beta v^{\beta-1}}\biggl(\log\int_\R
\exp\bigl(a\beta S(z)v^{\beta-1}-ar_2 v^{\beta-2}+q(v-S(z))\bigr)
\mu(\dd z)+r_1\biggr).
$$
Letting $v\to+\infty$, we obtain 
$$
S(u)=C\quad\mbox{for all $u\in\R$}.
$$
This contradicts the condition that~$S$ is non-constant and proves the strict positivity of the entropy production rate. 

\smallskip
We now turn to the general situation. The following proposition provides a sufficient condition for the positivity of the mean entropy production in the stationary regime.

\begin{proposition} \label{p2.11}
In addition to the hypotheses of Theorem~\ref{t2.9}, let us assume that the densities~$\tilde\rho_j$ are representable in the form
\begin{equation} \label{2.71}
\tilde\rho_j(y)=\exp\bigl(-a_j|r|^{\beta_j}+q_j(r)\bigr),\quad r\in\R,
\end{equation}
where $a_j>0$ and $\beta_j\in(1,2]$ are some numbers and~$q_j$ are continuously differentiable functions such that 
\begin{equation} \label{2.72}
a_j\le C, \quad \sup_{r\in\R}\bigl(|q_j(r)|+|q_j'(r)|\bigr)\le C,\quad j\ge1,
\end{equation}
where $C>0$ does not depend on~$j$. Then $\langle\sigma\rangle_\mu>0$. 
\end{proposition}

\begin{proof}
As in the case of the Burgers equation and Gaussian perturbations, it suffices to prove that the detailed balance~\eqref{DBCond} cannot hold $\ell\otimes\ell$ almost everywhere. We shall argue by contradiction. 

\smallskip
{\it Step~1: Continuity of transition densities\/}. 
We first show that the function~$\rho(u,v)$ is continuous on~$H\times H$. Indeed, relations~\eqref{8.4} and~\eqref{2.71} imply that 
\begin{equation} \label{2.73}
\rho(u,v)=\exp\biggl\{\sum_{j=1}^\infty P_j(u,v_j)\biggr\},
\end{equation}
where we set
$$
P_j(u,v_j)=-\frac{a_j}{|b_j|^{\beta_j}}\bigl(|v_j-S_j(u)|^{\beta_j}-|v_j|^{\beta_j}\bigr)
+q_j\bigl(\tfrac{v_j-S_j(u)}{b_j}\bigr)-q_j\bigl(\tfrac{v_j}{b_j}\bigr). 
$$
It follows from~\eqref{2.72} and the mean value theorem that
\begin{align*}
\bigl||v_j-S_j(u)|^{\beta_j}-|v_j|^{\beta_j}\bigr|&
\le 2|S_j(u)|\bigl(|v_j|^{\beta_j-1}+|S_j(u)|^{\beta_j-1}\bigr),\\
\bigl|q_j\bigl(\tfrac{v_j-S_j(u)}{b_j}\bigr)-q_j\bigl(\tfrac{v_j}{b_j}\bigr)\bigr|
&\le C|b_j|^{-1}|S_j(u)|. 
\end{align*}
Combining these estimates with the explicit formula for~$P_j$, applying the Cauchy--Schwarz inequality, and using~\eqref{2.035}, we derive
\begin{equation} \label{2.76}
\sum_{j=N}^\infty |P_j(u,v_j)|\le C_1(R+1)\biggl(\,\sum_{j=N}^\infty b_j^2\biggr)^{1/2}
\quad\mbox{for $u\in H$, $v\in B_H(R)$}. 
\end{equation}
Recalling~\eqref{2.4} and using the continuity of~$P_j(u,v_j)$ on the space~$H\times H$, we conclude that~$\rho(u,v)$ is continuous with respect to~$(u,v)\in H\times H$ and satisfies the inequality
\begin{equation} \label{2.74}
e^{-C_2(1+\|v\|)}\le\rho(u,v)\le e^{C_2(1+\|v\|)}. 
\end{equation}
It follows that the density~$\rho(v)$ of the stationary measure~$\mu$ is also continuous on~$H$ and satisfies the same inequality. 

\smallskip
{\it Step~2: Derivation of contradiction\/}. 
We now assume that detailed balance~\eqref{DBCond} holds $\ell\otimes\ell$ almost everywhere on~$H\times H$. Since the support of~$\ell\otimes\ell$ coincides with~$H\times H$ and all the functions entering relation~\eqref{DBCond} are continuous, we conclude  that it must be valid for all~$u,v\in H$. Taking the logarithm and using~\eqref{2.73}, we derive
\begin{multline} \label{2.75}
\sum_{j=1}^\infty P_j(u,v_j)=\log\int_H\exp\biggl\{\sum_{j=1}^\infty P_j(z,v_j)\biggr\}\mu(\dd z)\\
+\log\rho(v,u)-\log\int_H\rho(z,u)\mu(\dd z). 
\end{multline}
Let us fix a vector $u\in H$ and an integer $k\ge1$  and take $v_j=0$ for $j\ne k$ and $v_k=\lambda$ (with $\lambda\gg1$). In view of~\eqref{2.74},  the second and third terms on the right-hand side of~\eqref{2.75} remain bounded as $\lambda\to+\infty$. Using relation~\eqref{2.68}, we obtain
$$
-\beta_k\lambda^{\beta_k-1}S_k(u)
=\log\int_H\exp\bigl\{-\beta_k\lambda^{\beta_k-1}S_k(z)+r_1^\lambda(z)\bigr\}\mu(\dd z)+r_2^\lambda(u), 
$$
where we denote by~$r_i^\lambda(\cdot)$ functions that remain bounded as $\lambda\to+\infty$ uniformly with respect to the other variables. Dividing the above relation by $-\beta_k\lambda^{\beta_k-1}$ and letting $\lambda\to+\infty$, we conclude that $S_k(u)$ does not depend on~$u$ for any integer $k\ge1$. It follows that $S(u)$ is a constant mapping, which contradicts the backward uniqueness for problem~\eqref{2.51}, \eqref{2.52}; see\footnote{In~\cite{BV1992}, the backward uniqueness is proved for quasilinear parabolic equations in H\"older spaces. However, the same argument works also in Sobolev spaces.} Section~8 in~\cite[Chapter~2]{BV1992}. This completes the proof of Proposition~\ref{p2.11}. 
\end{proof}

\section{Exponential mixing and LDP}
\label{s2}

In this section, we prove Theorems~\ref{t1.1} and~\ref{t1.2}. To this end, we show that the Markov family in question satisfies the four hypotheses of Proposition~\ref{p10.3}, so that the LDP holds in the space of trajectories. We next use an approximation argument to establish the LDP for functionals with moderate growth at infinity. 

\subsection{Proof of Theorem~\ref{t1.1}}
\subsubsection*{Lyapunov function}
Let us show that~$\varPhi(u)$ satisfies~\eqref{10.10}. Indeed, in view of~\eqref{1.1}, we have 
$$
\int_H\varPhi(v)P_1(u,\dd v)=
\E\,\varPhi(S(u)+\eta_1)\le q\,\varPhi(u)+C\,\E\bigl(\varPhi(\eta_1)+1).
$$ 
This inequality coincides with~\eqref{10.10} in which $M=\E\,\varPhi(\eta_1)+C$, and the finiteness of~$M$ follows from~\eqref{1.07}. 

\subsubsection*{Uniform strong Feller}
We first note that $P_1(u,\cdot)=\ell_{S(u)}$. By Condition~(A), the mapping~$S$ is continuous from~$H$ to~$U$, and by Condition~(C), the mapping $\theta:U\to\PP(H)$ is continuous from~$U$ to~$\PP(H)$. We see that the mapping $u\mapsto P_1(u,\cdot)$ is continuous as the composition of two continuous mappings. 

\subsubsection*{Irreducibility}
By condition~(C), the support of~$\ell$ coincides with~$H$. Since the measure~$P_1(u,\cdot)$ is a translation of~$\ell$, the same property holds for it, and we see that $P_1(u,G)>0$ for any non-empty open set $G\subset H$. 

\subsubsection*{Super-exponential recurrence}
Since Borel measures on a Polish space are regular (e.g., see Theorem~1 in~\cite[Section~V.2]{GS1980}), given $\e>0$, we can find a compact subset~$\KK_\e\subset H$ such that $\ell(\KK_\e)>1-\e$. We claim that~\eqref{10.11} and~\eqref{5.31} hold for $\CC=B_U(R)+\KK_\e$ with $R\gg1$ and $\e\ll1$. The proof of this fact is divided into three steps. 

\medskip
{\it Step~1.} Let~$\sigma_\rho$ be the first hitting time of the set $\{u\in H:\varPhi(u)\le \rho\}$, which is denoted by~$\{\varPhi\le\rho\}$ in what follows, and let $\alpha=q\delta/C$, where the numbers~$q$, $C$, and~$\delta$ are defined in Conditions~(B) and~(C). As will be proved  in Step~3, for any $\beta>0$ there is~$\rho_0=\rho_0(\beta)>0$  such that
\begin{equation} \label{3.1}
\E_u e^{\beta\sigma_\rho}\le C_1 e^{\alpha\varPhi(u)-\alpha\rho+\beta}\quad
\mbox{for $u\in H$, $\rho\ge\rho_0$},
\end{equation}
where $C_1>0$ does not depend on~$\beta$, $\rho$, and~$u$. 
In this case, the validity of~\eqref{10.11} with the above choice of~$\CC$ can be derived by a standard argument (e.g., see Section~3.3.2 of~\cite{KS-book}). Indeed, choosing~$R$ so large that $S(\{\varPhi\le\rho\})\subset B_U(R)$, we see that
\begin{equation} \label{3.2}
\inf_{u\in \{\varPhi\le\rho\}}P_1(u,\CC)\ge 1-\e.
\end{equation}
Let us introduce an increasing sequence of stopping times by the relations
$$
\sigma_0'=\sigma_\rho, \quad 
\sigma_n'=\min\{k\ge \sigma_{n-1}'+1:\varPhi(u_k)\le\rho\}.
$$
Setting $\sigma_n=\sigma_n'+1$, we conclude from~\eqref{3.2} and the strong Markov property that 
\begin{equation} \label{3.3}
P_u(m):=\IP\biggl(\,\bigcap_{n=1}^m\{u_{\sigma_n}\notin \CC\}\biggr)\le \e^m. 
\end{equation}
We shall show in Step~2 that, for any $\beta>0$, there is $Q_\beta>1$ such that
\begin{equation} \label{3.6}
\E_ue^{\beta \sigma_m}\le C_2Q_\beta^me^{\alpha\varPhi(u)-\alpha\rho}
\quad\mbox{for $m\ge0$, $u\in H$},
\end{equation}
where $C_2>0$ is independent of~$\beta$, $\rho$, $m$, and~$u$. Using~\eqref{3.3}, \eqref{3.6}, and the Chebyshev inequality, for any positive integers~$m$ and~$M$ we write
\begin{align*}
\IP_u\{\tau_\CC\ge M\}
&=\IP_u\{\tau_\CC\ge M,\sigma_m<M\}+\IP_u\{\tau_\CC\ge M,\sigma_m\ge M\}\\
&\le \IP_u\{\tau_\CC>\sigma_m\}+\IP_u\{\sigma_m\ge M\}\\
&\le \IP_u\{u_{\sigma_1}\notin \CC,\dots,u_{\sigma_m}\notin \CC\}+e^{-\beta M}\,\E_ue^{\beta\sigma_m}\\
&\le \e^m+C_2Q_\beta^m e^{-\beta M+\alpha\varPhi(u)-\alpha\rho}. 
\end{align*}
Choosing $m$ to be the largest integer smaller than $\frac{M}{\log Q_\beta}$ and setting $\e=Q_\beta^{1-\beta}$, we derive
$$
\IP_u\{\tau_\CC\ge M\}\le C_3(\beta)\bigl(1+e^{\alpha\varPhi(u)-\alpha\rho}\bigr)e^{-(\beta-1)M},
$$
whence, for any $A<\beta-1$, it follows that
\begin{equation} \label{3.8}
\E_ue^{A\tau_\CC}\le C_4(\beta,A)\bigl(1+e^{\alpha\varPhi(u)-\alpha\rho}\bigr). 
\end{equation}
Since $\beta>0$ was arbitrary, we see that~\eqref{10.11} holds with any $A>0$ and a suitable compact set $\CC(A)\subset H$. Moreover, taking if necessary a larger constant~$C>0$ in inequality~\eqref{1.1}, we can make~$\alpha$ smaller than the number~$c>0$ entering~\eqref{1.8}. Then,  integrating~\eqref{3.8} with respect to~$\lambda(\dd u)$, we conclude that~\eqref{5.31} is also satisfied. 

\smallskip
{\it Step~2}. 
We now prove~\eqref{3.6}. To this end, we introduce the stopping time $\sigma_\rho'=\min\{k\ge1:\varPhi(u_k)\le\rho\}$. In view of~\eqref{3.1} and the Markov property, we have
$$
\E_u e^{\beta\sigma_\rho'}\le C_5 e^{\alpha\varPhi(u)-\alpha\rho+2\beta},
$$
where $C_5>0$ does not depend on the other parameters. 
Combining this inequality with the strong Markov property and the fact that $u_{\sigma_n'}\in \{\varPhi\le\rho\}$, we derive
\begin{align*}
\E_u e^{\beta\sigma_m'}
&=\E_u\bigl(\E_u\bigl\{e^{\beta\sigma_m'}\,|\,\FF_{\sigma_{m-1}'}\bigr\}\bigr)=\E_u\bigl(e^{\beta\sigma_{m-1}'}
\E_{u(\sigma_{m-1}')}e^{\beta\sigma_\rho'}\bigr)\\
&\le C_5e^{2\beta}\,\E_ue^{\beta\sigma_{m-1}'}, 
\end{align*}
where $\FF_\tau$ denotes the $\sigma$-algebra associated with the stopping time~$\tau$, and we write $u(\sigma_n')$ for~$u_{\sigma_n'}$. Iterating the above inequality and using the definition of~$\sigma_m$, we obtain the required estimate~\eqref{3.6}.

\smallskip
{\it Step~3}. 
It remains to prove inequality~\eqref{3.1}, in which $\rho\ge\rho_0$ with some constant~$\rho_0=\rho_0(\beta)>0$ chosen below. First note that, in view of the inequality $I_{\{\sigma_\rho>1\}}\le \exp(\delta'\varPhi(u_1)-\delta'\rho)$, where $\delta'>0$, and relations~\eqref{1.1} and~\eqref{1.07}, we have
\begin{align*}
\E_u\bigl(e^{\alpha\varPhi(u_1)}I_{\{\sigma_\rho>1\}}\bigr)
&\le e^{-\delta'\rho}\,\E_ue^{(\alpha+\delta')\varPhi(u_1)}\\
&\le e^{-\delta'\rho}\,\E_ue^{(\alpha+\delta')
\{q\varPhi(u)+C(\varPhi(\eta_1)+1)\}}\\
&\le e^{-\delta'\rho+(\alpha+\delta')(q\varPhi(u)+C)}\mmmm_{C(\alpha+\delta')},
\end{align*}
where $\mmmm_\delta(\ell)$ is defined in~\eqref{1.07}. 
Choosing $\delta'=(1-q)\delta/C$ and recalling that $\alpha=q\delta/C$, we obtain
\begin{equation} \label{3.4}
\E_u\bigl(e^{\alpha\varPhi(u_1)}I_{\{\sigma_\rho>1\}}\bigr)
\le C_3\mmmm_\delta(\ell)\,e^{\alpha\varPhi(u)-\delta'\rho}. 
\end{equation}
We now introduce the quantities 
$p_k(u)=\E_u(e^{\alpha\varPhi(u_k)}I_{\{\sigma_\rho>k\}})$. Combining~\eqref{3.4} with the Markov property, we obtain
\begin{align*}
p_{k+1}(u)
&=\E_u\bigl(e^{\alpha\varPhi(u_{k+1})}I_{\{\sigma_\rho>k+1\}}\bigr)
=\E_u\bigl\{I_{\{\sigma_\rho>k\}}\,\E_{u_k}\bigl(e^{\alpha\varPhi(u_1)}I_{\{\sigma_\rho>1\}}\bigr)\bigr\}\\
&\le C_3\mmmm_\delta(\ell)\,
\E_u\bigl(e^{\alpha \varPhi(u_k)-\delta'\rho}I_{\{\sigma_\rho>k\}}\bigr)
= C_3\mmmm_\delta(\ell)e^{-\delta'\rho}p_k(u).
\end{align*}
Iterating this inequality, using~\eqref{3.4}, and setting $\rho_0(\beta)=(C_4+\beta+1)/\delta'$, we get
$$
\E_u\bigl(e^{\alpha\varPhi(u_k)}I_{\{\sigma_\rho>k\}}\bigr)
\le e^{\alpha\varPhi(u)-(\delta'\rho-C_4)k}
\le e^{\alpha\varPhi(u)-(\beta+1)k},
$$
where $C_4=\log(C_3\mmmm_\delta(\ell))$ and~$\rho\ge\rho_0$. It follows that 
\begin{equation} \label{3.5}
\IP_u\{\sigma_\rho>k\}
\le e^{-\alpha\rho}\E_u\bigl(e^{\alpha\varPhi(u_k)}I_{\{\sigma_\rho>k\}}\bigr)
\le e^{\alpha\varPhi(u)-\alpha\rho-(\beta+1)k}. 
\end{equation}
Inequality~\eqref{3.1} with arbitrary~$\beta>0$ and $\rho\ge\rho_0(\beta)$ is a simple consequence of~\eqref{3.5}. 

\subsection{Proof of Theorem~\ref{t1.2}}
\medskip
{\it Step~1: Scheme of the proof of LDP}. 
We shall derive the LDP for the laws of~$\xi_k$ as a consequence of Theorem~\ref{t1.1}. To this end, we essentially repeat the argument used by Gourcy~\cite{gourcy-2007a,gourcy-2007b} in the case of the Navier--Stokes and Burgers equations. It is based on Lemma~2.1.4 of~\cite{DS1989}, which implies that the LDP with the rate function\footnote{We may consider~$f$ as a measurable function on~$\HHH$ depending only on the first $m+1$ components of the argument $\vvv=(v_n,n\ge0)$, so that the integral $\langle f,\nnu\rangle$ makes sense.}
\begin{equation} \label{4.0}
I_f(r):=\inf\{\III(\nnu):\langle f,\nnu\rangle=r\}
\end{equation}
will be established for the $\IP_\lambda$-laws of~$\xi_k$ if we prove the following two properties:
\begin{itemize}
\item[\bf(a)]
Let $f_j=(f\wedge j)\vee(-j)$ (where $a\wedge b$ and $a\vee b$ denote, respectively, the minimum and maximum of~$a$ and~$b$) and let $I_m:\PP(H^{m+1})\to[0,+\infty]$ be defined by 
$$
I_m(\nu)=\inf\{\III(\nnu):\nnu\in\PP(\HHH),\Pi_m\nnu=\nu\},
$$
where $\Pi_m:\HHH\to H^{m+1}$ denotes the natural projection sending the vector  $\vvv=(v_n,n\ge0)$ to~$(v_0,\dots,v_m)$ and $\III$ is the rate function constructed in Theorem~\ref{t1.1}. Then, for any $L>0$, we have
\begin{equation} \label{4.1}
\sup_{\nu}|\langle f_j- f,\nu\rangle|\to0\quad\mbox{as $j\to\infty$}, 
\end{equation}
where the supremum is taken over all $\nu\in\PP(H^{m+1})$ such that $I_m(\nu)\le L$.
\item[\bf(b)]
For any $\delta>0$, we have
\begin{equation} \label{4.2}
\limsup_{k\to+\infty}\frac1k\log\IP_\lambda\bigl\{\bigl|\bigl\langle f_j- f,\zeta_k^{(m)}\bigr\rangle\bigr|>\delta\bigr\}\to-\infty
\quad\mbox{as $j\to\infty$},
\end{equation}
where~$\zeta_k^{(m)}$ denote the occupation measures
\begin{equation} \label{3.7}
\zeta_k^{(m)}=\frac1k\sum_{n=0}^{k-1} \delta_{u_n(m)}, \quad 
u_n(m)=(u_n,\dots,u_{n+m}). 
\end{equation}
\end{itemize}

To prove~(a), we shall need the following lemma, which gives a lower bound for~$I_m$ in terms of the stabilisable functional~$\pppp$. Its proof is given at the end of this section. 

\begin{lemma} \label{l4.1}
Let the hypotheses of Theorem~\ref{t1.1} be fulfilled and let~$\pppp(u)$ be a uniformly stabilisable functional for~$(u_k,\IP_u)$. Then 
\begin{equation} \label{4.3}
I_m(\nu)\ge\frac{1}{m+1}
\int_{H^{m+1}}\sum_{n=0}^m\pppp(v_n)\,\nu(\dd v_0,\dots,\dd v_m)
-\gamma\quad\mbox{for any $\nu\in\PP(H^{m+1})$}.
\end{equation}
\end{lemma}

\smallskip
{\it Step~2: Proof of~(a)}.
We first note that, in view of~\eqref{4.3}, if $I_m(\nu)\le L$, then 
\begin{equation} \label{4.5}
\int_{H^{m+1}}\sum_{n=0}^m\pppp(v_n)\,\nu(\dd v_0,\dots,\dd v_m)
\le (m+1)(L+\gamma).
\end{equation}
Furthermore, since~$ f$ is bounded on the balls of~$H^{m+1}$, we have
\begin{equation} \label{4.6}
A_j:=\bigl\{(v_0,\dots,v_m)\in H^{m+1}:| f(v_0,\dots,v_m)|\ge j\bigr\}
\subset B_m(r_j)^c,\quad j\ge1,
\end{equation}
where $B_m(r)$ denotes the ball in~$H^{m+1}$ of radius~$r$ centred at zero and $\{r_j\}\subset\R_+$ is a sequence going to~$+\infty$ with~$j$. It follows from~\eqref{1.11} that 
\begin{equation} \label{4.7}
\frac{| f(v_0,\dots,v_m)|}{\pppp(v_0)+\cdots+\pppp(v_m)}\le \e_j\quad\mbox{for $(v_0,\dots,v_m)\in B_m(r_j)^c$},
\end{equation}
where $\e_j\to0$ as $j\to\infty$. Combining~\eqref{4.5}--\eqref{4.7}, we write
\begin{align*}
\int_{H^{m+1}}| f_j- f|\,\dd\nu
&= \int_{A_j}| f|\,\dd\nu\le \int_{B_m(r_j)^c}| f|\,\dd\nu\\
&\le \e_j\int_{H^{m+1}}\sum_{n=0}^m\pppp(v_n)\,
\nu(\dd v_0,\dots,\dd v_m)\\
&\le \e_j(m+1)(L+\gamma). 
\end{align*}
This implies the required convergence~\eqref{4.1}. 

\smallskip
{\it Step~3: Proof of~(b)}.
Using~\eqref{4.7}, we write
\begin{align}
\IP_\lambda\bigl\{\bigl|\bigl\langle f_j- f,\zeta_k^{(m)}\bigr\rangle\bigr|>\delta\bigr\}
&=\IP_\lambda\biggl\{\biggl|\frac1k\sum_{n=0}^{k-1}( f- f_j)(u_n,\dots,u_{n+m})\biggr|>\delta\biggr\}\notag\\
&=\IP_\lambda\biggl\{\sum_{n=0}^{k-1}| f(u_n,\dots,u_{n+m})|\,I_{A_j}(u_n,\dots,u_{n+m})>\delta k\biggr\}\notag\\
&\le\IP_\lambda\biggl\{\sum_{n=0}^{k+m-1}\pppp(u_n)>\frac{\delta}{(m+1)\e_j}k\biggr\}.\label{3.016}
\end{align}
Now note that, in view of~\eqref{1.9} and the Chebyshev inequality, the probability on the right-hand side of~\eqref{3.016} can be estimated by
$$
C_{m,\gamma}\exp\bigl\{-k\bigl(\tfrac{\delta}{(m+1)\e_j}-\gamma\bigr)\bigr\}\int_He^{\pppp(u)}Q(\|u\|)\,\lambda(\dd u),
$$
where $C_{m,\gamma}=e^{\gamma(m-1)}$. Substituting this into~\eqref{3.016}, using~\eqref{1.17}, and recalling that $\e_j\to0$ as $j\to\infty$, we obtain the required convergence~\eqref{4.2}. This completes the proof of Theorem~\ref{t1.2}. 

\begin{proof}[Proof of Lemma~\ref{l4.1}]
Let $\pppp_j:H\to\R_+$ be an increasing sequence of bounded continuous functions such that $\pppp_j(u)\to\pppp(u)$ for any $u\in H$. For instance, we can take
$$
\pppp_j(u)=j\wedge\inf_{v\in H}(\pppp(v)+j\,\|u-v\|). 
$$
In this case, $\pppp_j$ is a $j$-Lipschitz function, and the lower semicontinuity of~$\pppp$ implies that $\pppp_j\to\pppp$ point wise. 
By the Varadhan lemma (see~\cite[Section~4.3]{ADOZ00}), 
\begin{multline} \label{4.4}
\lim_{k\to\infty}\frac1k\log\E_u\exp\bigl(\tfrac{k}{m+1}\langle\pppp_j(u_0)+\cdots+\pppp_j(u_m),\zeta_k^{(m)}\rangle\bigr)\\
=\sup_{\nu\in\PP(H^{m+1})}\biggl(\frac{1}{m+1}\int_{H^{m+1}}\sum_{n=0}^m\pppp_j(v_n)\,\nu(\dd v_0,\dots,\dd v_m)-I_m(\nu)\biggr). 
\end{multline}
On the other hand, since $\pppp_j\le \pppp$ and $\pppp$ is a uniformly stabilisable functional,  in view of~\eqref{1.9} we have
\begin{align*}
\E_u\exp\bigl(\tfrac{k}{m+1}\bigl\langle\pppp_j(u_0)+\cdots+\pppp_j(u_m),\zeta_k^{(m)}\bigr\rangle\bigr)
&\le\E_u\exp\biggl(\,\sum_{n=0}^{k+m-1}\pppp(u_n)\biggr)\\
&\le Q(\|u\|)e^{\gamma(k+m-1)+\pppp(u)}. 
\end{align*}
Substituting this inequality into~\eqref{4.4}, for any $\nu\in\PP(H^{m+1})$ we obtain
$$
\frac{1}{m+1}\int_{H^{m+1}}\sum_{n=0}^m\pppp_j(v_n)\,\nu(\dd v_0,\dots,\dd v_m)-I_m(\nu)\le \gamma.
$$
The required inequality~\eqref{4.3} follows now from the Fatou lemma. 
\end{proof}

\section{Gallavotti--Cohen fluctuation theorem for the entropy production functional}
\label{s11}

In this section, we prove Theorem~\ref{t1.5}. To this end, we first note that,  by Theorem~\ref{t1.2}, the $\IP_\lambda$-laws of the random variables~\eqref{1.14} satisfy the LDP with the good rate function (cf.~\eqref{4.0})
\begin{equation} \label{4.001}
I(r):=\inf\{\III(\nnu):\langle \sigma,\nnu\rangle=r\}. 
\end{equation}
We shall use the Varadhan lemma (see~\cite[Section~4.3]{ADOZ00}) and a symmetry property of the Feynman--Kac semigroup to prove the Gallavotti--Cohen fluctuation principle. 

\smallskip
For any $\alpha\in\R$ we define (formally) a family of linear mappings by the relation
$$
(\PPPP_k^{\alpha \sigma}f)(u)=\E_u\biggl\{\exp\biggl(-\alpha\sum_{n=0}^{k-1}\sigma(u_n,u_{n+1})\biggr)f(u_k)\biggr\}, \quad k\ge0,
$$
where $f\in C(H)$. In view of Condition~(D) and inequality~\eqref{1.9}, the function $\PPPP_k^{\alpha \sigma}f$ is continuous on~$H$ for any $f\in C_b(H)$. We claim that\,\footnote{The right- and left-hand sides of~\eqref{4.9} are well defined as integrals of positive functions. Relation~\eqref{4.9} means, in particular, that if one of them is infinite, then so is the other.}
\begin{equation} \label{4.9}
(\PPPP_k^{\alpha\sigma}f,g)_\ell=(f,\PPPP_k^{(1-\alpha)\sigma}g)_\ell
\quad\mbox{for all $k\ge0$},
\end{equation}
where $f,g\in C(H)$ are arbitrary non-negative functions, and we set
$$
(f,g)_\ell=\int_Hf(u)g(u)\,\ell(\dd u). 
$$
Indeed, the Markov property for~$(u_k,\IP_u)$ implies that, for any continuous function $f\ge0$, we have
$$
\PPPP_{k+l}^{\alpha\sigma}f=\PPPP_k^{\alpha\sigma}(\PPPP_l^{\alpha\sigma}f),\quad k,l\ge0. 
$$
Therefore it suffices to prove~\eqref{4.9} for $k=1$. Using the definitions of~$\sigma$ and~$\PPPP_1^{\alpha\sigma}$ and the Fubini theorem, we write
\begin{align*}
(\PPPP_1^{\alpha\sigma}f,g)_\ell
&= \int_H\E_u\bigl\{e^{-\alpha\sigma(u_0,u_1)}f(u_1)\bigr\}g(u)\,\ell(\dd u)\\
&=\int_H\biggl(\int_He^{\alpha(\log \rho(v,u)-\log\rho(u,v))}f(v)\rho(u,v)\ell(\dd v)\biggr)g(u)\ell(\dd u)\\
&=\int_Hf(v)\biggl(\int_He^{(1-\alpha)(\log \rho(u,v)-\log\rho(v,u))}g(u)\rho(v,u)\ell(\dd u)\biggr)\ell(\dd v)\\
&=\int_Hf(v)\E_v\bigl\{e^{-(1-\alpha)\sigma(u_0,u_1)}g(u_1)\bigr\}\ell(\dd v)
=(f,\PPPP_1^{(1-\alpha)\sigma}g)_\ell. 
\end{align*}

\smallskip
We can now derive~\eqref{1.15}. Since the measure~$\ell$ satisfies~\eqref{1.17}, the LDP with the good rate function~\eqref{4.001} holds for the $\IP_\ell$-laws of the sequence of random variables~$\xi_k$ defined by~\eqref{1.14}. We claim that the Varadhan lemma is applicable to~$\{\alpha\xi_k\}$ for any $\alpha\in\R$. To this end, it suffices to check that 
\begin{equation} \label{4.11}
\limsup_{k\to\infty}\frac{1}{k}\log\E_\ell\exp(\beta k|\xi_k|)<\infty,
\end{equation}
where $\beta>0$ is arbitrary. To see this, let us note that, in view of~\eqref{1.13} and the boundedness of~$\sigma$ on any ball of~$H\times H$, we have
$$
k|\xi_k|\le \e\sum_{n=0}^k\pppp(u_n)+kR_\e,\quad k\ge1,
$$
where $\e>0$ is arbitrary, and~$R_\e>0$ depends only~$\e$. This inequality combined with~\eqref{1.9} and~\eqref{1.17} (for $\lambda=\ell$) implies the validity of~\eqref{4.11} for any $\beta>0$. 

Hence, applying the Varadhan lemma, we see that the following limit exists and is finite for any $\alpha\in\R$:
\begin{equation} \label{4.10}
\lim_{k\to\infty}\frac1k\log\E_\ell\exp(-\alpha k\xi_k)
=\sup_{r\in \R}\bigl(-\alpha r-I(r)\bigr)=:I^*(-\alpha), 
\end{equation}
where~$I^*$ stands for the Legendre transform of~$I$. Now note that
$$
\E_\ell\exp(-\alpha k\xi_k)
=\bigl(\PPPP_k^{\alpha\sigma}{\mathbf1},{\mathbf1}\bigr)_\ell, 
$$
where $\mathbf1:H\to\R$ stands for the function identically equal to~$1$. Substituting this relation into~\eqref{4.10} and using~\eqref{4.9}, for any $\alpha\in\R$ we derive
\begin{align*}
I^*(-\alpha)&=\lim_{k\to\infty}
\frac1k\log\bigl(\PPPP_k^{\alpha\sigma}{\mathbf1},{\mathbf1}\bigr)_\ell
=\lim_{k\to\infty}
\frac1k\log\bigl(\PPPP_k^{(1-\alpha)\sigma}{\mathbf1},{\mathbf1}\bigr)_\ell
=I^*(\alpha-1). 
\end{align*}
Combining this with the well-known relation $I(r)=\sup_{\alpha\in\R}(\alpha r-I^*(\alpha))$, we obtain
$$
I(-r)=\sup_{\alpha\in\R}\bigl(-\alpha r-I^*(\alpha)\bigr)
=\sup_{\alpha\in\R}\bigl(\alpha r-I^*(-\alpha)\bigr)
=\sup_{\alpha\in\R}\bigl(\alpha r-I^*(\alpha-1)\bigr)=I(r)+r.
$$ 
This completes the proof of Theorem~\ref{t1.5}. 

\section{Appendix}
\label{s8}
\subsection{Admissible shifts of decomposable measures}
\label{s8.1}
Let~$H$ be a separable Hilbert space endowed with its Borel $\sigma$-algebra~$\BB_H$. Given~$\mu\in\PP(H)$ and~$a\in H$, we denote by~$\theta_a:H\to H$ the shift operator by the vector~$a$ (that is, $\theta_au=u+a$) and by~$\mu_a=\mu\circ\theta_a^{-1}$ the image of~$\mu$ under~$\theta_a$. Recall that $a\in H$ is called an {\it admissible shift\/} for~$\mu$ if $\mu_a$ is absolutely continuous with respect to~$\mu$. We denote by~$H_\mu$ the set of all admissible shifts for~$\mu$ and by $\rho_\mu(a;u)=\frac{\dd\mu_a}{\dd\mu}$ the corresponding densities. It is straightforward to check that~$H_\mu$ is an additive semigroup in~$H$.

We shall say that~$\mu$ is a {\it decomposable measure\/} if there is an orthonormal basis~$\{e_j\}$ in~$H$ such that 
\begin{equation} \label{8.1}
\mu=\bigotimes_{j=1}^\infty \mu_j,
\end{equation}
where $\mu_j=\mu\circ {\mathsf P}_j^{-1}$, and ${\mathsf P}_j:H\to H$ is the orthogonal projection to the straight line spanned by~$e_j$. It is clear that if~$\mu$ is a decomposable measure, then it is the law of a random variable of the form\,\footnote{For instance, one can take the random variables $\xi_j=(u,e_j)_H$ on the probability space $(H,\BB_H,\mu)$.}
\begin{equation} \label{8.2}
\eta=\sum_{j=1}^\infty \xi_je_j,
\end{equation}
where $\{\xi_j\}$ is a sequence of independent scalar random variables such that $\DD(\xi_j)=\mu_j$. A proof of the following result can be found in~\cite{GS1980} (see Theorem~5 in Section~VII.2). 

\begin{proposition} \label{p8.1}
Let~$\mu$ be a decomposable measure such that~$\mu_j$ possesses a density~$\rho_j$ with respect to the Lebesgue measure on~$\R$ for any $j\ge1$. Then $a\in H_\mu$ if and only if the series
\begin{equation} \label{8.3}
\sum_{j=1}^\infty\bigl(\log \rho_j\bigl(\xi_j-(a,e_j)_H\bigr)-\log \rho_j(\xi_j)\bigr)
\end{equation}
converges almost surely. In this case, the corresponding density is given by 
\begin{equation} \label{8.4}
\rho_\mu(a;u)=\exp\biggl(\sum_{j=1}^\infty\log \frac{\rho_j(u_j-a_j)}{\rho_j(u_j)}\biggr), \quad u\in H,
\end{equation}
where we set $u_j=(u,e_j)_H$ and $a_j=(a,e_j)_H$.
\end{proposition}

Let us note that, in view of the Kolmogorov zero-one law, series~\eqref{8.3} either converges a.s.\;or diverges a.s. In the latter case, the measures~$\mu_a$ and~$\mu$ are mutually singular. Furthermore, if $a\in H_\mu$, then $\mu\ll\mu_a$. What has been said implies that, under the hypotheses of the proposition, the subset $H_\mu\subset H$ is a group, and the measures~$\mu_a$ and~$\mu_{a'}$ with $a,a'\in H_\mu$ are absolutely continuous with respect to each other, with the corresponding density given by
\begin{equation} \label{8.5}
\rho_\mu(a,a';u)=\frac{\dd\mu_a}{\dd\mu_{a'}}(u)
=\exp\biggl(\sum_{j=1}^\infty\log \frac{\rho_j(u_j-a_j)}{\rho_j(u_j-a_j')}\biggr), \quad u\in H. 
\end{equation}

We now wish to estimate the total variation distance between two admissible shifts of a decomposable measure. To this end, we assume that~$\mu_j$ is the law of a random variable of the form $\xi_j=b_j\tilde\xi_j$, where $\{b_j\}$ is a sequence of positive numbers and~$\tilde\xi_j$ is a random variable whose law is absolutely continuous with respect to the Lebesgue measure, and the corresponding density~$\tilde\rho_j\in C^1$ is positive everywhere and satisfies the inequality
\begin{equation} \label{8.6}
\var(\tilde\rho_j)\le C\quad\mbox{for all $j\ge1$},
\end{equation}
where $\var(\cdot)$ denotes the total variation of a function and~$C>0$ does not depend on~$j$. 

\begin{proposition} \label{p8.2}
Let~$\mu$ be a decomposable measure satisfying the above hypotheses. Then for any $a,a'\in H_\mu$ we have 
\begin{equation} \label{8.7}
\|\mu_{a}-\mu_{a'}\|_{\mathrm{var}}\le \frac{C}{2}\sum_{j=1}^\infty\frac{|a_j-a_j'|}{b_j},
\end{equation}
where~$C$ is the same constant as in~\eqref{8.6}. 
\end{proposition}

\begin{proof}
Let us recall that
\begin{equation} \label{8.8}
\|\mu_{a}-\mu_{a'}\|_{\mathrm{var}}=\frac12\int_{H}
|\rho_\mu(a;u)-\rho_\mu(a';u)|\,\mu(\dd u).
\end{equation}
In view of~\eqref{8.4}, we have
$$
\rho_\mu(a;u)-\rho_\mu(a';u)
=\sum_{k=1}^\infty D_k(a,a';u)
\frac{\rho_k(u_k-a_k)-\rho_k(u_k-a_k')}{\rho_k(u_k)},
$$
where 
$$
D_k(a,a';u)=\exp\biggl(\sum_{j=1}^{k-1}\log\frac{\rho_j(u_j-a_j)}{\rho_j(u_j)}
+\sum_{j=k+1}^\infty\log\frac{\rho_j(u_j-a_j')}{\rho_j(u_j)}\biggr). 
$$
Substituting the above relation into~\eqref{8.8} and using decomposition~\eqref{8.1}, we obtain
\begin{align} 
&\|\mu_{a}-\mu_{a'}\|_{\mathrm{var}}
\le\frac12\sum_{k=1}^\infty\int_H|D_k(a,a';u)|
\frac{|\rho_k(u_k-a_k)-\rho_k(u_k-a_k')|}{\rho_k(u_k)}\mu(\dd u)
\notag\\
&=\frac12
\sum_{k=1}^\infty\biggl\{\prod_{j\ne k}\int_\R\frac{\rho_j(u_j-a_{jk})}{\rho_j(u_j)}\mu_j(\dd u_j)\biggr\} 
\int_\R\frac{|\rho_k(u_k-a_k)-\rho_k(u_k-a_k')|}{\rho_k(u_k)}\mu_k(\dd u_k)
\notag\\
&=\frac12
\sum_{k=1}^\infty \int_\R|\rho_k(u_k-a_k)-\rho_k(u_k-a_k')|\dd u_k,
\label{8.9}
\end{align}
where $a_{jk}=a_j$ for $j<k$ and $a_{jk}=a_j'$ for $j>k$. The mean value theorem implies that
$$
\rho_k(u_k-a_k)-\rho_k(u_k-a_k')
=\int_0^1\rho_k'\bigl(u_k-\theta a_k-(1-\theta)a_k'\bigr)d\theta\,(a_k'-a_k).
$$
Combining this with~\eqref{8.9}, recalling that $\mu_j(\dd u_j)=\rho_j(u_j)\,\dd u_j$, and using the relation $\rho_k(r)=b_k^{-1}\tilde\rho_k(b_k^{-1}r)$, we obtain~\eqref{8.7}.
\end{proof}

Finally, we need some sufficient conditions ensuring the continuity of the shift operator and the positivity of density for shifted measures. 

\begin{proposition} \label{p5.3}
Let $\mu$ be the same as in Proposition~\ref{p8.2} and let the densities~$\tilde\rho_j$ satisfy the inequalities
\begin{equation} \label{10.15}
\int_\R|r|\tilde \rho_j(r)\,\dd r\le C_1, \quad 
\int_{-A}^A\frac{|\tilde \rho_j'(y-r)|}{\tilde \rho_j(y-r)}\,\dd r\le C_1(|y|+1)A
\end{equation}
for $j\ge1$, $y\in \R$, and $A\in[0,1]$. Let $U\subset H$ be a Banach space such that 
\begin{equation} \label{10.13}
\sum_{j=1}^\infty b_j^{-1}|(v,e_j)|
\le C_2\|v\|_U\quad \mbox{for any $v\in U$},
\end{equation}
where~$\{e_j\}$ is the orthonormal basis entering~\eqref{8.1}.
Then the inclusion $U\subset H_\mu$ holds, and the density~$\rho_\mu(a;u)$ is positive for~$(a,u)\in U\times U$. Moreover, the function $\theta:U\to \PP(H)$ taking $a\in U$ to~$\mu_a$ is Lipschitz continuous, provided that~$\PP(H)$ is endowed with the total variation norm. Finally, if there is $C_3>0$ such that 
\begin{equation} \label{10.18}
\sum_{j=1}^\infty b_j^{-2}|(v,e_j)|
\le C_3\|v\|_U\quad \mbox{for any $v\in U$},
\end{equation}
then the density $\rho_\mu(a;u)$ is positive on $U\times H$. 
\end{proposition}

\begin{proof}
We first note that if $a,a'\in H_\mu$, then inequalities~\eqref{8.7} and~\eqref{10.13} imply that
$$
\|\mu_a-\mu_{a'}\|_{\mathrm{var}}\le \frac{C}{2}\sum_{j=1}^\infty b_j^{-1}|(a-a',e_j)|\le\frac{C_2C}{2}\|a-a'\|_U,
$$
whence we conclude that $\theta:U\to\PP(H)$ is Lipschitz continuous. Thus, we need to show the inclusion $U\subset H_\mu$ and the positivity of~$\rho_\mu$ on~$U\times U$ (and on $U\times H$ under the additional condition~\eqref{10.18}). 

In view of Proposition~\ref{p8.1}, the required inclusion will be established if we prove that
\begin{equation} \label{10.14}
\sum_{j=N}^\infty\int_H\bigl|\log\rho_j(u_j-a_j)-\log\rho_j(u_j)\bigr|\,\mu(\dd u)<\infty
\quad\mbox{for any $a\in U$},
\end{equation}
where $N\ge1$ is an integer depending on~$a$.
To prove this, note that, in view of the second inequality in~\eqref{10.15}, we have
\begin{align}
\bigl|\log\rho_j(u_j-a_j)-\log\rho_j(u_j)\bigr|
&\le \int_{-\theta_j}^{\theta_j}
\frac{|\tilde \rho_j'(b_j^{-1}u_j-r)|}{\tilde \rho_j(b_j^{-1}u_j-r)}\,\dd r\notag\\
&\le C_1(b_j^{-1}|u_j|+1)\theta_j
\quad\mbox{for $j\ge N$},
\label{10.16}
\end{align}
where $\theta_j=b_j^{-1}|a_j|$ and $N\ge1$ is the least integer such that $\theta_j\le 1$. Using the Fubini theorem and decomposition~\eqref{8.1}, we obtain
\begin{align*}
\sum_{j=N}^\infty
\int_H\bigl|\log\rho_j(u_j-a_j)-\log\rho_j(u_j)\bigr|\,\mu(\dd u)
&\le C_1\sum_{j=N}^\infty \theta_j \int_\R (b_j^{-1}|u_j|+1)\rho_j(u_j)\,\dd u_j\\
& =C_1\sum_{j=N}^\infty \theta_j \int_\R (|v_j|+1)\tilde\rho_j(v_j)\,\dd v_j. 
\end{align*}
The first inequality in~\eqref{10.15} and inequality~\eqref{10.13} with $v=a$ now imply that~\eqref{10.14} holds.  

\smallskip
To establish the positivity of the density~$\rho_\mu(a;u)$ on $U\times U$, recall that it is given by~\eqref{8.4}. Therefore it suffices to show that 
\begin{equation} \label{10.17}
\Delta_N(a;u):=\sum_{j=N}^\infty\bigl|\log\rho_j(u_j-a_j)-\log\rho_j(u_j)\bigr|<\infty
\end{equation}
for $a,u\in U$. To this end, note that, by~\eqref{10.16} and~\eqref{10.13}, we have
\begin{equation} \label{10.19}
\Delta_N(a;u)\le C_1\sum_{j=N}^\infty
\bigl(b_j^{-2}|u_j|\,|a_j|+b_j^{-1}|a_j|\bigr)
\le C_1C_2\|a\|_U\bigl(\|u\|_U+1\bigr). 
\end{equation}

Finally, to establish the positivity of~$\rho_\mu(a;u)$ on~$U\times H$ under the additional condition~\eqref{10.18}, it suffices to prove that~\eqref{10.17} holds for $a\in U$ and $u\in H$. This follows immediately from the first inequality in~\eqref{10.19}. 
\end{proof}

\subsection{Exponential mixing and LDP for Markov chains}
\label{s8.2}
Let~$X$ be a separable Banach space with a norm~$\|\cdot\|$ and let~$(u_k,\IP_u)$ be a family of Markov chains in~$X$. Given $\lambda\in\PP(X)$, we define the probability measure $\IP_\lambda(\cdot)=\int_X\IP_u(\cdot)\lambda(\dd u)$ and denote by~$\E_\lambda$ the corresponding mean value. Recall that we denote by $P_k(u,\Gamma)$ the transition function for $(u_k,\IP_u)$ and by~$\PPPP_k$ and~$\PPPP_k^*$ the corresponding Markov semigroups. Given a closed subset $K\subset X$, let~$\tau_K$ be the first positive hitting time of~$K$: 
$$
\tau_K=\min\{k\ge1:u_k\in K\}. 
$$
The following proposition is a consequence of general results on mixing and LDP established in Theorem~2.1 and Proposition~A.2 of~\cite{wu-2001}; see also~\cite[Chapters~15 and~16]{MT1993} for some results on mixing under more general hypotheses. 

\begin{proposition} \label{p10.3}
Let a Markov process $(u_k,\IP_u)$ and a subset $\Lambda\subset\PP(X)$ be such that the following hypotheses hold.

\smallskip
\noindent
{\bf Lyapunov function.}
There is a continuous function  $\varPhi:X\to\R_+$ which is bounded on any ball of~$X$ and goes to~$+\infty$ as $\|u\|\to\infty$ such that
\begin{equation} \label{10.10}
\int_X\varPhi(v)P_1(u,\dd v)\le q\,\varPhi(u)+M
\quad\mbox{for all $u\in X$},
\end{equation}
where $q<1$ and $M$ are some positive constants. 

\smallskip
\noindent
{\bf Uniform strong Feller.}
The mapping $u\mapsto P_1(u,\cdot)$ is continuous from~$X$ to the space~$\PP(X)$ endowed with the total variation norm. 

\smallskip
\noindent
{\bf Irreducibility.}
We have $P_1(u,G)>0$ for any $u\in X$ and any non-empty open set $G\subset X$. 

\smallskip
\noindent
{\bf Hyper-exponential recurrence.}
For any $A>0$ there is a compact subset $\CC=\CC(A)\subset X$ such that 
\begin{align} 
\sup_{u\in B}\E_u\exp(A\tau_{\CC})&<\infty\quad\mbox{for every ball $B\subset X$}, \label{10.11}\\
\sup_{\lambda\in \Lambda}\E_\lambda 
\exp\bigl(A\tau_{\CC}\bigr)&<\infty.
\label{5.31}
\end{align}
Then $(u_k,\IP_u)$ possesses a unique stationary measure $\mu\in\PP(X)$, which is exponentially mixing in the sense that
\begin{equation} \label{10.12}
\|P_k(u,\cdot)-\mu\|_{\mathrm{var}}
\le Ce^{-\gamma k}\bigl(1+\varPhi(u)\bigr)
\quad\mbox{for all $u\in X$ and $k\ge0$},
\end{equation}
where $C$ and~$\gamma$ are positive constants. Moreover, the LDP in the $\tau_p$-topology holds for the $\IP_\lambda$-occupation measures~\eqref{1.12}, uniformly with respect to~$\lambda\in \Lambda$.
\end{proposition}

This result implies, in particular, that for any ball $B\subset X$ the LDP holds for the $\IP_u$-occupation measures~\eqref{1.12} uniformly with respect to $u\in B$. Let us also note that the hyper-exponential recurrence is needed only for the LDP: the uniqueness of a stationary measure and exponential convergence to it remain valid if we require that~\eqref{10.11} is valid for a fixed compact set $\CC\subset X$ and a number $A>0$.

\subsection{The entropy balance equation}
\label{s5.3}
In this section, we consider the entropy balance equation stated in the 
introduction. We prove that the entropy production functional 
$\Ep(\,\cdot\,)$ is 
non-negative
and vanishes if and only if  the detailed balance condition is satisfied. We also show
how the entropy flux observable $\sigma$ relates to time-irreversibility.

Let us set $f=\dd\lambda/\dd\ell$ and define the density transfer operator by 
$$
(\mathcal{R}f)(v)=\int f(u)\rho(u,v)\ell(\dd u).
$$ 
Denoting by~$\rho$ the density of the stationary distribution~$\mu$ with respect to~$\ell$ and using the relations
$$
\frac{\dd\lambda}{\dd\mu}(u)=\frac{f(u)}{\rho(u)}, \quad \frac{\dd\lambda_1}{\dd\mu}(u)=\frac{(\RR f)(u)}{\rho(u)}, 
$$
it is straightforward to show that~\eqref{EntBalance} holds with
\begin{equation}\label{EpForm}
\Ep(\lambda)=\int_\HHH-\log\left(
\frac{(\mathcal{R}f)(u_1)\rho(u_1,u_0)}{f(u_0)\rho(u_0,u_1)}\right)
\boldsymbol{\lambda}(\dd\uuu).
\end{equation}
Jensen's inequality yields
\begin{align*}
\Ep(\lambda)&\ge-\log\int_\HHH
\frac{(\mathcal{R}f)(u_1)\rho(u_1,u_0)}{f(u_0)\rho(u_0,u_1)}
\boldsymbol{\lambda}(\dd\uuu)\\
&=-\log\int_{H\times H} f(u)\rho(u,v)\frac{(\mathcal{R}f)(v)\rho(v,u)}{f(u)\rho(u,v)}
\ell(\dd u)\ell(\dd v)\\
&=-\log\int_H(\mathcal{R}^2f)(u)\ell(\dd u)=0.
\end{align*}
Moreover, this inequality is saturated if and only if  
$$
(\mathcal{R}f)(v)\rho(v,u)=c f(u)\rho(u,v),
$$
$\ell\otimes\ell$-almost everywhere for some constant $c$. Integrating this
relation over~$\ell(\dd v)$ yields $\mathcal{R}^2f=cf$, and one more integration over~$\ell$ shows that $c=1$. We deduce that $\mathcal{R}f-f$ is either $0$ 
or an eigenvector of $\mathcal{R}$ corresponding to the eigenvalue $-1$. Since the
second alternative contradicts the mixing property of the stationary measure 
$\mu$, we conclude that $\mathcal{R}f=f$, so~$f=\rho$ is the density of a stationary measure. Inserting this relation into Eq.~\eqref{EpForm} yields
\begin{align*}
0=\Ep(\lambda)&=\int_{H\times H} \rho(u)\rho(u,v)\log\left(
\frac{\rho(u)\rho(u,v)}{f(v)\rho(v,u)}\right)\ell(\dd u)\ell(\dd v)\\
&=\frac12\int_{H\times H}\left(\rho(u)\rho(u,v)-\rho(v)\rho(v,u)\right)\log\left(
\frac{\rho(u)\rho(u,v)}{\rho(v)\rho(v,u)}\right)\ell(\dd u)\ell(\dd v).
\end{align*}
Recalling that the logarithm is a strictly increasing function, we see that the expression under the last integral is nonnegative, and the integral vanishes if and only if $\rho(u)\rho(u,v)=\rho(v)\rho(v,u)$ for $\ell\otimes\ell$-almost every $(u,v)$. Thus, the detailed balance relation must hold.

To connect the observables $J$ and $\sigma$ with time reversal of the stationary Markov chain, we follow
Maes and Neto{\v{c}}n{\'y}~\cite{MN-2003}.
Denote by $\mathbb{P}_\mu^{(k)}$ the measure induced by $\mu$ on the
finite segment $(u_0,\ldots,u_{k})$ of the Markov chain and by 
$\pi_k:(u_0,\ldots,u_{k})\mapsto(u_{k},\ldots,u_0)$
the time-reversal map on this segment. The relative entropy
 $\Ent(\mathbb{P}_\mu^{(k)}|\mathbb{P}_\mu^{(k)}\circ\pi_k^{-1})$ is
given by
\begin{align*}
&-\int\log\left(\frac{\dd\mathbb{P}_\mu^{(k)}}
{\dd \mathbb{P}_\mu^{(k)}\circ\pi_k^{-1}}\right)
\dd\mathbb{P}_\mu^{(k)}\\
=&-\int\log\left(
\frac{\rho(u_0)\rho(u_0,u_1)\cdots\rho(u_{k-1},u_{k})}
{\rho(u_{k})\rho(u_{k},u_{k-1})\cdots\rho(u_{1},u_{0})}\right)
\mathbb{P}_\mu^{(k)}(\dd u_0,\ldots,\dd u_{k})\\
&=-\int\sum_{n=0}^{k-1}J\circ\phi^n(\uuu)\boldsymbol{\mu}(\dd\uuu),
\end{align*}
where $\phi$ denotes the left shift. Thus, we have
\begin{equation}
\label{NMRel}
-\frac1k\Ent(\mathbb{P}_\mu^{(k)}|\mathbb{P}_\mu^{(k)}\circ\pi_k^{-1})
=\langle\sigma\rangle_\mu,
\end{equation}
which provides an alternative proof of the inequality $\langle\sigma\rangle_\mu\ge0$.
As noticed by Gaspard~\cite{gaspard-2004}, the ergodicity of $\boldsymbol{\mu}$
implies that
\begin{align*}
\lim_{k\to\infty}\frac1k\log&\left(\rho(u_0)\rho(u_0,u_1)\cdots\rho(u_{k-1},u_{k})\right)\\
&=\int\rho(u)\rho(u,v)\log\rho(u,v)\ell(\dd u)\ell(\dd v)=h_+,\\
\lim_{k\to\infty}\frac1k\log&\left(\rho(u_k)\rho(u_k,u_{k-1})\cdots
\rho(u_1,u_0)\right)\\
&=\int\rho(u)\rho(u,v)\log\rho(v,u)\ell(\dd u)\ell(\dd v)=h_-,
\end{align*}
for $\boldsymbol{\mu}$-almost every $\uuu\in\HHH$, where $h_+$ is the entropy per unit time
(or entropy rate, or Kolmogorov--Sinai entropy) of the stationary Markov chain and $h_-$ is the entropy per unit time of the time-reversed process. Thus, relation~\eqref{NMRel} can be strengthened to
$$
-\lim_{k\to\infty}\frac1k\log\left(\frac{\dd\mathbb{P}_\mu^{(k)}}
{\dd \mathbb{P}_\mu^{(k)}\circ\pi_k^{-1}}(\uuu)\right)=h_--h_+
=\langle\sigma\rangle_\mu,
$$
and the strict positivity of the entropy production rate translates into
$h_->h_+$.

\addcontentsline{toc}{section}{Bibliography}
\def\cprime{$'$} \def\cprime{$'$}
  \def\polhk#1{\setbox0=\hbox{#1}{\ooalign{\hidewidth
  \lower1.5ex\hbox{`}\hidewidth\crcr\unhbox0}}}
  \def\polhk#1{\setbox0=\hbox{#1}{\ooalign{\hidewidth
  \lower1.5ex\hbox{`}\hidewidth\crcr\unhbox0}}}
  \def\polhk#1{\setbox0=\hbox{#1}{\ooalign{\hidewidth
  \lower1.5ex\hbox{`}\hidewidth\crcr\unhbox0}}} \def\cprime{$'$}
  \def\polhk#1{\setbox0=\hbox{#1}{\ooalign{\hidewidth
  \lower1.5ex\hbox{`}\hidewidth\crcr\unhbox0}}} \def\cprime{$'$}
  \def\cprime{$'$} \def\cprime{$'$} \def\cprime{$'$}
\providecommand{\bysame}{\leavevmode\hbox to3em{\hrulefill}\thinspace}
\providecommand{\MR}{\relax\ifhmode\unskip\space\fi MR }
\providecommand{\MRhref}[2]{%
  \href{http://www.ams.org/mathscinet-getitem?mr=#1}{#2}
}
\providecommand{\href}[2]{#2}

\end{document}